\documentclass[aps,prx,onecolumn,numerical,10pt]{revtex4-2}
\usepackage{natbib}
\usepackage[english]{babel}
\usepackage{amsmath}
\usepackage{amsthm}
\usepackage{amssymb}
\usepackage{mathtools}
\usepackage{ragged2e}
\usepackage{etoolbox}
\usepackage{booktabs}
\usepackage{graphicx}

\usepackage{tikz-cd} 

\usepackage{chessfss}
\usepackage{color}

\usepackage{cases}
\usepackage{enumerate}
\usepackage{enumitem}
\usepackage{multirow}

\usepackage{natbib}
\usepackage{pgfplots}
\pgfplotsset{compat=1.18}
\usepackage{url}
\newcommand{\ket}[1]{ | #1 \rangle }
\newcommand{\bra}[1]{ \langle #1 | }

\newcommand{\tr}{\operatorname{tr}}
\newcommand{\F}{ \mathbb F}

\theoremstyle{definition}
\newtheorem{theorem}{Theorem}[section]
\newtheorem{definition}[theorem]{Definition}
\newtheorem{remark}[theorem]{Remark}
\newtheorem{example}[theorem]{Example}
\newtheorem{lemma}[theorem]{Lemma}
\newtheorem{corollary}[theorem]{Corollary}
\newtheorem{proposition}[theorem]{Proposition}

\usepackage{scalerel,stackengine}
\stackMath
\newcommand\reallywidehat[1]{%
\savestack{\tmpbox}{\stretchto{%
  \scaleto{%
    \scalerel*[\widthof{\ensuremath{#1}}]{\kern-.6pt\bigwedge\kern-.6pt}%
    {\rule[-\textheight/2]{1ex}{\textheight}}
  }{\textheight}%
}{0.5ex}}%
\stackon[1pt]{#1}{\tmpbox}%
}
\usepackage{comment}

\usepackage[margin=3cm]{geometry}
\usepackage{hyperref}
\hypersetup{
  colorlinks   = true,  
  urlcolor     = black,  
  linkcolor    = blue,  
  citecolor    = red    
}

\begin{document}


\title{CSS-T Codes over Binary Extension Fields and Their Physical Foundations}
\author{Jasper J. Postema, Fabrizio Conca and Alberto Ravagnani}
\affiliation{Eindhoven University of Technology, the Netherlands}

\maketitle

\section*{Abstract}
We investigate the class of CSS-$T$ codes, a family of quantum error-correcting codes that allows for a transversal $T$-gate. We extend the definition of a pair of linear codes $(C_1,C_2)$, $C_i\subseteq\mathbb{F}_q^n$, forming a $q$-ary CSS-$T$ code over binary extension fields, and demonstrate the existence of asymptotically good sequences of LDPC CSS-$T$ codes over any such field.

\section*{Introduction}

Quantum computers hold the potential to execute certain algorithms more efficiently than classical computers, such as Shor's algorithm for prime factorisation or Grover's search algorithm \cite{primefac1,primefac2,grover}. However, these devices are highly susceptible to noise and errors. Various error mitigation protocols have been proposed \cite{zne, vqocsse, virtualdistillation, digitalzne, review}, but quantum error correction is still regarded as necessary to attain high-fidelity quantum computing in the future \cite{gottesmanphd, errorfigure, errorfigure2}. 

In 1995, Calderbank and Shor \cite{CS}, and independently Steane \cite{S}, proved that quantum error-correcting codes do exist. Their construction, known as the Calderbank-Shor-Steane (CSS) codes, derives a quantum code from two classical linear codes $C_2\subseteq C_1\subseteq\mathbb{F}_q^n$. Since then, many quantum codes have been studied, often derived from classical linear codes. Some quantum constructions have even been introduced as self-contained ansätze, without explicitly relying on underlying classical codes \cite{bicycleansatz,generalisedbicycleansatz, bravyi}.

Quantum computers implement logic gates through unitary operations, which can be classified according to the Clifford hierarchy. Most CSS codes allow transversal gates only from the first two levels of this hierarchy, though a gate from the third level is required to achieve a universal gate set, i.e., a set of unitaries from which any arbitrary unitary gate can be derived \cite{universalgateset}. The Eastin-Knill Theorem, however, prevents any error-correcting stabiliser code from transversally implementing a universal gate set \cite{eastinknill}, i.e., in a fault-tolerant manner that avoids propagating errors across qubits. Since most gates are non-Clifford, this motivates the search for codes with a transversal non-Clifford gate, such as a $T$-gate.

Recently, the notion of a CSS-$T$ code
was proposed in~\cite{csstoriginal,css-t}. A CSS-$T$ code is a binary CSS code that allows the $T$-gate to be executed transversally, enabling fault-tolerant implementation of logical non-Clifford gates and reducing overhead in quantum computation. It has been an open question whether an asymptotically good family of CSS-$T$ codes exists, as in the case of CSS codes~\cite{proofasymp}. This question was answered in \cite{goodcsst}, where the authors showed that
a binary CSS code can be transformed into a CSS-$T$ code of double the length, and used this fact
to prove the existence of asymptotically good sequences of (LDPC) binary CSS-$T$ codes \cite{goodcsst}.

The definition of a CSS-$T$ code proposed in~\cite{csstoriginal, css-t} and studied in~\cite{goodcsst} relies heavily on the base field being $\F_2$, and it is not clear how to extend it to larger finite fields. An attempt was made in~\cite{alberto}, where a definition of a $q$-ary CSS-$T$ code was proposed and investigated from a mathematical viewpoint, though without offering a physical motivation.

In this paper, we propose a physically grounded definition of
CSS-$T$ codes over any binary field extension $\mathbb{F}_{2^s}$, which differs from the one in~\cite{alberto}. Our definition is inspired by the presence of the field trace in the definition of a $q$-ary $T$-gate. We then study the fundamental properties of CSS-$T$ codes over binary field extensions, and show both differences and analogies with the binary case. We also extend the approach of~\cite{goodcsst}, and show that
there exist
asymptotically good sequences of LDPC CSS-$T$ codes over any field extension of the form $\mathbb{F}_{2^s}$.

The remainder of this paper is structured as follows: In Section~\ref{sec:qc}, we provide an introduction to quantum computing. Classical and quantum error-correcting codes are introduced in Section~\ref{sec:css}, along with CSS-$T$ codes over $\mathbb{F}_2$ and an algebraic characterisation. We use this in Section~\ref{sec:reedmuller} to revisit constructions of
CSS-$T$ codes using Reed-Muller codes efficiently. In Section~\ref{sec:qarycsst}, we propose a new definition of CSS-$T$ codes
over binary field extensions
$\mathbb{F}_{2^s}$. Some constructions are shown in Section~\ref{sec:applications}, including an example of CSS-$T$ codes derived from cyclic codes. In Section~\ref{sec:bounds}, we study the asymptotic behaviour of long CSS-$T$ codes and prove that asymptotically good CSS-$T$~codes exist over any field of characteristic $2$.

\section{Quantum computing}\label{sec:qc}

This section contains a brief self-contained introduction to quantum computing and establishes the notation for the rest of the paper. Let $q= p^s$ be a prime power and $\mathbb{F}_q$ the finite field with $q$ elements. Let $\mathbb{C}$ denote the complex field, and $\dagger$ Hermitian conjugation. A ket-vector $\ket{\cdot}$ denotes a column vector, while a bra-vector $\bra{\cdot}=\ket{\cdot}^\dagger$ denotes a row vector. Let $\{\ket{u_0},\ldots, \ket{u_{q-1}}\}$ denote an orthonormal basis 
    \begin{equation*}
        \ket{u_0} = \begin{pmatrix}1\\0\\\vdots\\0\end{pmatrix}, \quad \ket{u_1} = \begin{pmatrix}0\\1\\\vdots\\0\end{pmatrix}, \quad\cdots,\quad \ket{u_{q-1}} = \begin{pmatrix}0\\0\\\vdots\\1\end{pmatrix}
    \end{equation*}
    for the Hilbert space $\mathfrak{H}_q\cong\mathbb{C}^{\otimes q}$, expressed as
    \begin{equation*}
        \mathfrak{H}_q = \bigg\{\sum_{i=0}^{q-1} \alpha_i \ket{u_i} \mid \alpha_0, \ldots, \alpha_{q-1} \in \mathbb{C} \bigg\}.
    \end{equation*}
    This Hilbert space is equipped with the standard inner product
    \begin{equation*}
        \langle\alpha\, | \,\beta\rangle = \sum_{i=0}^{q-1} \overline{\alpha}_i \beta_i \quad \text{ for all } \alpha, \beta \in  \mathfrak{H}_q, 
    \end{equation*}
    which also defines the norm of any vector:
    \[v=\sum_{i=0}^{q-1} \alpha_i \ket{u_i}, \quad \| \ket{v} \| = \sqrt{\langle v \,|\, v\rangle} = \sqrt{\sum_{i=0}^{q-1}{|\alpha_i|^2} }.\]
    
\begin{definition}[Qudits]
    A \textit{qudit} is an element $\ket{Q}\in\mathfrak{H}_q$ with norm $\|\ket{Q}\| = 1$. If $q=2$, a qudit is called a \textit{qubit}. For any qudit $\ket{Q}=\sum_{i=0}^{q-1} Q_i \ket{u_i}$, the  $Q_i$'s are called the \textit{probability amplitudes}.
\end{definition}

From the point of view of quantum information theory, any quantum system described by a qudit $\ket{Q}$ has a probability $|Q_i|^2$ to collapse into the $\ket{u_i}$-state upon measurement in the appropriate basis (of which the so-called $Z$-basis is the standard), according to the \textit{Born rule}.

For qubits (i.e., when $q=2$), any single-qubit unitary gate can be decomposed into a complete $2\times 2$-basis called the \textit{Pauli basis}, consisting of the following matrices
\begin{equation*}
    I=\begin{pmatrix}
        1 & 0\\
        0 & 1
    \end{pmatrix}, \quad
    X = \begin{pmatrix}
        0 & 1\\
        1 & 0
    \end{pmatrix},\quad
    Y = \begin{pmatrix}
        0 & -i\\
        i & 0
    \end{pmatrix},\quad
    Z = \begin{pmatrix}
        1 & 0\\
        0 & -1
    \end{pmatrix},
\end{equation*}
which satisfy the commutation relations 
\[[X,Y]=2iZ, \quad [Y,Z] = 2iX, \quad [X,Z]=-2iY.\]

To generalize the Pauli group to qudits, we first need a definition that will prove crucial throughout the remainder of this paper as well.

\begin{definition}
The \emph{absolute trace map} of the field extension $\mathbb F_q / \mathbb F_p$, $q=p^s$, is 
    \begin{equation*}
    \operatorname{tr}:\mathbb F_q\to\mathbb F_p, \quad \operatorname{tr}(x) = \sum_{i=0}^{s-1} x^{p^i}.
\end{equation*}
\end{definition}
Pauli operators admit the following generalisation for any qudit defined over a finite field. For $\lambda\in\mathbb{F}_q$, let
\begin{equation}\label{eq:ztrace}
    X^{(\lambda)} = \sum_{x\in\mathbb{F}_q} \ket{x+\lambda}\bra{x},\quad Z^{(\lambda)}=\sum_{x\in\mathbb{F}_q} \zeta^{\text{tr}(\lambda x)}\ket{x}\bra{x},
\end{equation}
where $\zeta = \text{exp}\left(\frac{2\pi i}{p}\right)$ is a $p$-th root of unity. The presence of the trace map in the definition of the $q$-ary $Z$ operators will be crucial in the following sections of this paper. 
It can easily be shown that for any $\mu,\nu\in\mathbb{F}_q$,
\begin{equation*}
    X^{(\mu)}X^{(\nu)}=X^{(\mu+\nu)}\quad\text{and}\quad Z^{(\mu)}Z^{(\nu)}=Z^{(\mu+\nu)}.
\end{equation*}
The group generated by 
\[
\{1,\zeta,\ldots,\zeta^{p-1} \}\cdot\{ X^{(\lambda)}, Z^{(\lambda)} \mid \lambda \in \mathbb F_q \}
\]
is called the \emph{$q$-ary Pauli group}. The elements of this group can be written in Weyl-Heisenberg representation as follows.

\begin{definition}[Weyl-Heisenberg representation]
    For $a,b\in\mathbb{F}_q^n$, we let    \begin{equation*}
        E(a,b)=\sqrt{{\zeta}^{(a,b)}}X^aZ ^b,
    \end{equation*}
    where $X^a=X^{(a_1)}\otimes\cdots\otimes X^{(a_n)}$, and  for $Z^b=Z^{(b_1)}\otimes\cdots\otimes Z^{(b_n)}$.
\end{definition}

The elements of the $q$-ary Pauli group  operate independently only on single qudits, but the generation of \textit{entanglement} among qudits is a crucial component of any useful quantum algorithm. A two-qudit system $\ket{\psi}$ is called \textit{separable} if it can be written in the form
\begin{equation*}
    \ket{\psi}=\ket{\varphi_1}\otimes\ket{\varphi_2},
\end{equation*}
otherwise it is called \textit{entangled}. Examples of entangled qubits are the four Bell states:
\begin{equation*}
    \ket{\Phi^\pm} = \frac{\ket{00}\pm\ket{11}}{\sqrt{2}} \quad\text{and}\quad \ket{\Psi^\pm} = \frac{\ket{01}\pm\ket{10}}{\sqrt{2}}.
\end{equation*}

Gates that can entangle multiple qudits can be defined iteratively from the Pauli group.

\begin{definition}[Clifford hierarchy]
    The \emph{Clifford hierarchy} is the
    nested sequence of
    subsets of unitary operators defined recursively as follows:
    \begin{itemize}
        \item the first level $\mathcal K^{(1)}$ is taken to be the $q$-ary Pauli group,
        \item for any $n\geq 1$ we define the $(n+1)$-th level to be the normaliser of the $n$-th level, i.e.
    \begin{equation}\label{eq:clifford}
        \mathcal{K}^{(n+1)} = \{U \text{ unitary} \mid UPU^\dagger\in\mathcal{K}^{(n)} \text{ for all }P\in\mathcal{K}^{(n)}\}.
    \end{equation}
    \end{itemize}
    The first level $\mathcal{K}^{(1)}$ is the Pauli group, while the second level $\mathcal{K}^{(2)}$ is called the \textit{Clifford group}. Note that, for $n\geq 3$, $\mathcal{K}^{(n)}$ is no longer a group, although in $\mathcal{K}^{(3)}$ the subset of diagonal operators does form a group still.
\end{definition}

Any quantum circuit that only employs unitary gates from the Pauli and Clifford groups can be efficiently simulated classically in polynomial time, according to the Gottesman-Knill Theorem~\cite{gottesmanknill}. A gate set with only unitary gates from those groups cannot constitute a \textit{universal gate set}, i.e. a set of gates such that any unitary can be decomposed into elements from that basis set. In particular, for $q=2$ the Solovay-Kitaev Theorem says that if any gate set is dense in $SU(2)$, it can approximate any unitary gate with a low-depth quantum circuit~\cite{solovay}. Thus, we need to supply the Pauli and Clifford groups with a unitary from the third level of the Clifford hierarchy to build a universal gate set.
\begin{example}[Universal gate set~\cite{universalgateset}]
   The set $\{\text{CNOT}, H, T\}$ is a universal gate set, whose gates are given by 
   \begin{equation*}
       \text{CNOT} = \begin{pmatrix}
           1 & 0 & 0 & 0\\
           0 & 1 & 0 & 0\\
           0 & 0 & 0 & 1\\
           0 & 0 & 1 & 0
       \end{pmatrix},\quad H =\frac{1}{\sqrt{2}}(X+Z)=\frac{1}{\sqrt{2}}\begin{pmatrix}1 & 1\\1&-1\end{pmatrix},\quad T = \begin{pmatrix}1 & 0\\0 & e^{i\pi/4}\end{pmatrix}.
   \end{equation*}
   The $T$-gate is a diagonal operator from $\mathcal{K}^{(3)}$, and therefore supplements the Clifford group $\{\text{CNOT},H\}$ to be universal.
\end{example}

\begin{example}[Universal gate set~\cite{toffolihadamard}]
    Another example of a universal gate set is $\{H, \text{CCZ}\}$. For qubits, the latter's matrix representation is
    \begin{equation*}
        \text{CCZ} = \begin{pmatrix}1 & 0 & 0 & 0 & 0 & 0 & 0 &0\\
        0 & 1 & 0 & 0 & 0 & 0 & 0 &0\\
        0 & 0 & 1 & 0 & 0 & 0 & 0 &0\\
        0 & 0 & 0 & 1 & 0 & 0 & 0 &0\\
        0 & 0 & 0 & 0 & 1 & 0 & 0 &0\\
        0 & 0 & 0 & 0 & 0 & 1 & 0 &0\\
        0 & 0 & 0 & 0 & 0 & 0 & 1 & 0\\
        0 & 0 & 0 & 0 & 0 & 0 & 0 & -1\end{pmatrix}.
    \end{equation*}
\end{example}

In the $q$-ary case, we can generalize any quantum gate located above the Pauli set in the Clifford hierarchy; see Definition~\ref{eq:clifford}. Examples include the $q$-ary $T$-gate and $q$-ary $\text{CCZ}$-gate defined as follows for any $\lambda\in\mathbb{F}_q$:
\begin{equation}\label{eq:ttrace}
        T^{(\lambda)} = \sum_{x\in\mathbb{F}_q} e^{\frac{i\pi}{4}\text{tr}(\lambda x)}\ket{x}\bra{x}\quad\text{and}\quad \text{CCZ}^{(\lambda)} = \sum_{x,y,z\in\mathbb{F}_q} \zeta^{\text{tr}(\lambda xyz)}\ket{x}\ket{y}\ket{z}\bra{x}\bra{y}\bra{z},
\end{equation}
with $\zeta$ a $p$-th root of unity. We once again highlight the appearance of the field trace, which will play a predominant role in this paper. If the superscript is omitted, it is understood that we take $\lambda=1$, i.e. $T=T^{{(1)}}$.

\section{The CSS and CSS-\texorpdfstring{$T$}{T} constructions} \label{sec:css}

CSS codes, named after Calderbank, Shor, and Steane \cite{CS,S}, are a family of quantum error-correcting codes constructed from a pair of classical linear codes. Calderbank and Shor, and independently Steane, introduced this class of codes in 1996. Their constructions are different, but equivalent. CSS codes have been mostly studied in the binary case, but we work with an arbitrary field size $q$ throughout the paper. 

We start with the general definition of quantum code. In analogy with the classical setting, one  defines a \textit{quantum encoding} as a mapping from a certain Hilbert space to a larger-dimensional one.

\begin{definition}[Quantum encoding]\label{def:encoding}
    A \emph{quantum encoding} is an injective and norm-preserving linear operator: $\Phi:\mathfrak{H}_q^{\otimes k}\to\mathfrak{H}_q^{\otimes n}$. A \emph{quantum error-correcting code} $\mathcal{Q} = \text{im}(\Phi)$ is the image of such a mapping, and it has \emph{length} $n$ and \emph{dimension} $k$. Its distance $d$ is equal to the minimum Hamming weight of any non-zero vector in $\mathcal{Q}$. Such a $\mathcal{Q}$ is referred to as a (\emph{quantum}) $[[n,k,d]]_q$-\textit{code}.
\end{definition}

An encoding $\Phi$ also determines what gates on the \textit{physical level} (i.e., unitary operations of the form $U^{\otimes n}$) correspond to a logical gate (i.e., unitary operations of the form $U_L^{\otimes k}$). An operator $\mathcal{O}$ is said to be \emph{transversal} if $\mathcal{O}^{\otimes n}$ preserves the code space. Transversal gates can be implemented on quantum hardware with very low overhead and are fault-tolerant in the sense that they cannot multiply errors among large patches of qudits. 

We now turn to CSS codes. To define them, we need some concepts from classical coding theory, which we briefly review.  

\begin{definition}
A linear code over $\mathbb{F}_q$ is an $\mathbb F_q$-linear subspace $C\subseteq\mathbb{F}_q^n$, where $n\in\mathbb{N}$ is called the code \textit{length}. The \textit{dimension} of 
$C$ is its dimension as an  $\mathbb F_q$-linear subspace of $\mathbb{F}_q^n$, often denoted by $k$. The \textit{Hamming distance} between vectors $a,b \in \mathbb F_q^n$ is defined as
\begin{equation*}
    d^{\text H}(a,b) = |\{i\in\{1,\ldots, n\}\mid a_i \neq b_i\}|.
\end{equation*}
The distance of a vector $a\in C$ from the zero vector \textbf{0} is its \textit{Hamming weight}, $\omega^\text{H}(a)=d^\text{H}(a,\textbf{0})$, which counts 
the number of nonzero entries in $a$. The minimum distance $d$ of a nonzero linear code $C$ is 
\begin{equation*}
    d(C) = \min\{\omega^\text{H}(c) \mid c\in C, \, c\neq \textbf{0}\}.
\end{equation*}
A linear code $C \subseteq \mathbb F_q^n$ having length $n$, dimension $k$ and minimum distance $d$ is called an $[n,k,d]_q$-code. The dual of a code $C$ is the vector space of  vectors that are orthogonal to $C$ with respect to the standard inner product $\langle\cdot , \cdot\rangle$, namely:
\begin{equation*}
    C^\perp = \{x\in\mathbb{F}_q^n \mid \langle c,x\rangle = 0 \text{ for all } c\in C\}.
\end{equation*}
A code $C$ is called \textit{self-orthogonal} if $C\subseteq C^\perp$, and \emph{self-dual} if $C=C^\perp$. 
\end{definition}

We are now ready to define CSS codes.

\begin{definition}[CSS code]
    Let $C_2\subseteq C_1 \subseteq \mathbb{F}_q^n$ be linear codes. Let $\smash{\zeta = \exp\left(\frac{2\pi i}{p}\right)}$ be a primitive complex $p$-th root of unity and let $\text{tr}: \mathbb{F}_q \to \mathbb{F}_p$ be the trace map. For any vector $w\in\mathbb{F}_q^n$, define the qudit state
    \begin{equation*}
        \ket{c_w} = \frac{1}{\sqrt{|C_1|}}\sum_{c\in C_1}\zeta^{\text{tr} \langle c,w \rangle}\ket{c}.
    \end{equation*}
    The \textit{Calderbank-Shor quantum code} associated with the pair $(C_1, C_2)$ is
    \begin{equation*}
        Q^\text{CS}(C_1, C_2) = \{\ket{c_w}\mid w\in C_2^\perp\},
    \end{equation*}
    and the \textit{Steane quantum code} is given by
    \begin{equation*}
        Q^\text{S}(C_1, C_2) = \{\ket{w+C_2} \mid w\in C_1\},
    \end{equation*}
    where we let $\smash{\ket{w+C_2}=\frac{1}{\sqrt{|C_2|}}\sum_{c\in C_2} \ket{w+c}}$. 
\end{definition}

It can be shown that $Q^\text{CS}(C_1, C_2) = Q^\text{S}(C_2^\perp, C_1^\perp)$, illustrating that the two constructions are equivalent. CSS codes fall under the umbrella of \textit{stabiliser codes} \cite{gottesmanphd}, with the specific constraint that stabilisers contain only $X$-type Pauli operators, or only $Z$-like Pauli operators. These codes are quantum codes in the sense of the definition of a quantum encoding according to Definition \ref{def:encoding}. Clearly, they linearly transform a string of $k$ qudits into a string of $n$ qudits, where $n$ is the code length of $C_1$ and $C_2$, and $k=\text{dim}(C_1)-\text{dim}(C_2)$; see~\cite{nielsenchuang}. Additionally, the set $\{\ket{c_w}\}$ forms an orthonormal basis \cite{nielsenchuang}, proving that CSS quantum codes are an explicit example of the quantum encoding defined in Definition~\ref{def:encoding}.

Examples of CSS codes that have received a lot of attention in the quantum computing community are \textit{topological codes}, such as the \textit{surface code}~\cite{gameofsurfacecodes}, \textit{toric code}~\cite{toric}, \textit{colour codes}~\cite{colorcodes}, and the recently proposed class of \textit{bivariate bicycle codes}, with a promising low-density parity check (LDPC) behaviour~\cite{bravyi}.

The set of transversal gates that a CSS code can implement is limited. In fact, CSS codes
can never implement a transversal gate set, as the following result shows. 

\begin{theorem}[Eastin-Knill Theorem~\cite{eastinknill}] \label{theorem:eastin-knill}
    No quantum error-correcting stabiliser code $C$ can satisfy the following two properties simultaneously:
    \begin{itemize}
        \item the code distance of $C$ is greater than 2;
        \item the code allows for a transversal universal gate set.
    \end{itemize}
\end{theorem}

The CSS construction inherits the following transversal logical operators under the conditions specified in parentheses:
\begin{itemize}
    \item Pauli gates $I, X, Y, Z$ (always inherited for CSS codes);
    \item CNOT gate (always inherited for CSS codes);
    \item Hadamard gate $H$ (if and only if the CSS code is symmetric in its code generators, i.e., $C_2^\perp = C_1$).
\end{itemize}

A natural question is whether there exist CSS codes that transversally implement a $T$-gate. The family of binary CSS-$T$ codes, introduced in~\cite{csstoriginal}, possesses this feature. 



\begin{definition}[Binary CSS-$T$ code~\cite{csstoriginal}] \label{def:bincsst}
    A pair of linear codes $(C_1, C_2)$ with $C_2\subseteq C_1 \subseteq \mathbb{F}_2^n$ is called a \textit{CSS-$T$} pair if:
    \begin{itemize}
         \item $C_2$ is an even code, i.e., for all $x\in C_2$ we have $\sum_j x_j = 0$;
        \item for every $x\in C_2$ there exists a self-dual code $C_x \subseteq C_1^\perp$ of dimension $\omega^\text{H}(x)/2$ and supported on $x$, i.e., each $y \in C_x$ has $y_i=0$ whenever $x_i=0$.
    \end{itemize}
\end{definition}


Note that not all codes necessarily contain a self-dual code.
The following criterion tells us exactly when this happens over finite fields of characteristic 2.

\begin{lemma}[see~\cite{alberto}]\label{dualcontaining}
    Let $C\subseteq \mathbb F_q^n$ be a code over a binary extension field. Then $C$ contains a self-dual code if and only if its length $n$ is even and $C^\perp\subseteq C$, i.e.
    $C^\perp$ is self-orthogonal.
\end{lemma}

Through \textit{magic state distillation} \cite{msd}, we can non-transversally implement non-Clifford gates, and one should expect an overhead of the order of
\begin{equation*}
    \mathcal{O}\left(\log^\gamma\left(\frac{1}{\varepsilon}\right)\right),
\end{equation*}
where $\varepsilon$ is the accuracy of the distillation and $\gamma = \log(n/k)/\log(d)$ is the overhead constant~\cite{bravyihaahoverhead}, where $[[n,k,d]]$ again refer to the relevant code parameters. Instead of considering a single code to assess the resources needed for overhead, suppose there exists an \textit{asymptotically good} sequence  $\{C^\text{CSS}_i\}_{i\in\mathbb{N}}$ such that $\lim_{i\to\infty} n_i = \infty$ and is asymptotically good, i.e. it attains non-zero rate and relative minimum distance:
\begin{equation*}
    \limsup_{i\to\infty} \frac{k_i}{n_i}>0, \quad \limsup_{i\to\infty} \frac{d_i}{n_i}>0.
\end{equation*}
Then for sufficiently large $i$, the overhead can be made arbitrary small. More formally, for every $\varepsilon>0$ there exists an $i'$ such that $\gamma < \varepsilon$ whenever $i>i'$. Such a sequence achieves an asymptotically constant overhead. This observation strongly motivates
the search for code families that are
asymptotically good \textit{and} simultaneously transversally implement a $T$-gate.


\section{Revisiting CSS-\texorpdfstring{$T$}{T} codes from Reed-Muller codes}\label{sec:reedmuller}

In this short section we establish a characterisation of CSS-$T$ codes and show how it can be conveniently used  to derive one of the main results of~\cite{felice} with a short proofs. As we will see in Section~\ref{sec:qarycsst}, our characterisation has a natural extension to binary field extensions, while the original definition of CSS-$T$ codes does not.

\begin{definition}[Star product]
    The \emph{star product}, sometimes called the \emph{Schur product} or the \textit{elementwise product}, is the bilinear map  $\star: \mathbb F_q^n\times \mathbb F_q^n\to\mathbb{F}_q^n$ defined as
    \begin{equation*}
        a \star b \coloneqq (a_1 b_1,\ldots, a_n b_n)  
    \end{equation*}
    for all $a,b\in\mathbb{F}_q^n$. Given linear codes $A,B\subseteq \mathbb{F}_q^n$, their star product is defined as
    \[A\star B = \text{span}\{a\star b \mid a\in A, \, b\in B\}.\] 
    In the sequel, we abuse notation 
    and write $x\star C$ instead of $\langle x\rangle \star C$.
\end{definition}

\begin{remark}
\begin{enumerate}
    \item The $q$-ary repetition code $\mathcal{R}_q^n=\langle (1, 1, \ldots, 1)\rangle\subset\mathbb{F}_q^n$ is neutral with respect to the
    star product construction, in the sense that for all $C\subseteq\mathbb{F}_q^n$ we have $\mathcal{R}_q^n \star C = C$. Exponentiation with respect to the star product is iteratively defined as 
\begin{equation*} 
C^{\star t} = C\star C^{\star (t-1)},
\end{equation*}
with $C^{\star 0}= \mathcal{R}_q^n$ and $C^{\star 1}=C$. Note that for a binary code $C$, we always have $C\subseteq C^{\star 2}$, as $x=x\star x$ for any $x\in\mathbb{F}_2^n$.

\item The property of a binary code being even (i.e. every codeword has an even Hamming weight) can be characterised as follows: a binary code $C\subseteq\mathbb{F}_2^n$ is even if and only if $\mathcal{R}_2^n\subseteq C^\perp$.
\end{enumerate}
\end{remark}

The star product exhibits cyclical behaviour with respect to the standard inner product, in the following  sense.

\begin{lemma}[Cyclicity of the star product]
    For all $a,b,c\in\mathbb{F}_q$ we have
    \begin{equation*}
        \langle a\star b,c \rangle = \langle b\star c,a \rangle = \langle c\star a, b\rangle.
    \end{equation*}
    Moreover, for all linear codes $A,B,C\subseteq \mathbb{F}_q^n$ we have
    \begin{equation*}
        A\star B \subseteq C^\perp \Longleftrightarrow A \star C \subseteq B^\perp.
    \end{equation*}
\end{lemma}
\begin{proof}
The first part of the statement follows from the definitions.
Given linear codes $A,B,C\subseteq \mathbb{F}_q^n$ and $a\in A$, $b\in B$, $c\in C$, we have that $A\star B \subseteq C^\perp$ implies 
    \begin{equation*}
        0 = \langle a\star b,c \rangle = \langle c\star a, b \rangle
    \end{equation*}
    by the first part of the statement. Hence $A\star C \subseteq B^\perp.$ The other direction is analogous.
\end{proof}



We will need the following description of the CSS-$T$ property with respect to the star product.

\begin{theorem}[\text{\cite[Theorem 2.3]{eduardo}}]  \label{eduardo_csst}
   A binary CSS pair $(C_1, C_2)$ is CSS-$T$
   if and only if
    \[C_2 \subseteq C_1\cap (C_1^{\star 2})^\perp.\]
\end{theorem}


In this paper, we prove and apply a different characterization of the CSS-$T$ property via the star product,
which extends to larger fields of characteristic $2$ in a very natural way.

\begin{theorem}[Characterization of CSS-$T$ pairs over $\mathbb{F}_2$]\label{theorem:binarycsst}\label{CSST_cond}
    A binary CSS pair $(C_1, C_2)$ is CSS-$T$ if and only if \[C_1\star C_1 \subseteq C_2^\perp. \]
\end{theorem}
\begin{proof}
    Let $(C_1, C_2)$ be a binary CSS-$T$ pair. For every $x\in C_2$, there exists a self-dual code $C_x\subseteq C_1^\perp$ supported on $x$ and dimension $\omega^{\text H}(x)/2$. Pick arbitrary codewords $a\in C_1$, $x\in C_2$, and $z\in C_1$. Since $C_x^\perp=C_x \subseteq C_1^\perp$, we have $z\in C_x$. Therefore,
    \[ \langle a\star x,z \rangle = \langle a,x\star z \rangle  = \langle a,z\rangle  = 0.\]
    We conclude that $C_1\star C_2\subseteq C_1^\perp$. 
    For the other direction, suppose that $C_1\star C_1 \subseteq C_2^\perp$, and recall that by assumption $C_2\subseteq C_1$. Then $C_2 \subseteq C_1 \cap (C_1\star C_1)^\perp$, and $(C_1, C_2)$ is a CSS-$T$ pair by Theorem~\ref{eduardo_csst}.
\end{proof}


In the remainder of this section we 
prove one of the main results of 
\cite{felice} in a concise way using the characterization we provided in Theorem~\ref{CSST_cond}. In the statement of the next result, $\text{RM}(r,m)$ denotes the (binary) Reed-Muller code with parameters $(r,m)$; see~\cite{felice}.

\begin{theorem}[\text{\cite[Theorem 13]{felice}}]
    Let $C_1=\textnormal{RM}(\lfloor\frac{m-1}{2}\rfloor-t,m)$ and $C_2=\textnormal{RM}(r_2,m)$. Then, $\textnormal{CSS}(C_1,C_2)$ is a \textnormal{CSS-$T$} code if and only if
    \[\begin{cases}
        r_2\leq 2t+1 & \textnormal{if } m \textnormal{ is even},\\
        r_2\leq 2t& \textnormal{if } m \textnormal{ is odd}.
    \end{cases}\]
\end{theorem}
\begin{proof}
We start by recalling the inclusion properties of Reed-Muller codes  with respect to their first argument, i.e., 
\[\mbox{$\text{RM}(r_1,m)\subseteq\text{RM}(r_2,m)$ if and only if $r_1\leq r_2$},\]
and are closed under the action of the star product:
\begin{equation*}\mbox{$\textnormal{RM}(r_1,m)\star\textnormal{RM}(r_2,m)=\textnormal{RM}(r_1+r_2,m)$ for $m\geq 2$.}\end{equation*}
Using the latter fact 
and Theorem~\ref{theorem:binarycsst} we see that
$(C_1,C_2)$ is a CSS-$T$ pair if and only if
\begin{equation*}
\textnormal{RM}\left(2\bigg\lfloor\frac{m-1}{2}\bigg\rfloor-2t,m\right)\subseteq\textnormal{RM}(m-r_2-1,m).
    \end{equation*}
    Using the nested property of Reed-Muller codes, we conclude the following: if $m$ is even, then $\lfloor\frac{m-1}{2}\rfloor=\frac{m}{2}-1$, hence $r_2\leq 2t+1.$ If $m$ is odd, then $\lfloor\frac{m-1}{2}\rfloor=\frac{m}{2}-\frac{1}{2}$ and $r_2\leq 2t$.
\end{proof}

\section{\texorpdfstring{$q$}{q}-ary CSS-\texorpdfstring{$T$}{T} codes}\label{sec:qarycsst}

The original definition of a CSS-$T$ code
was given only over the binary field; recall Definition~\ref{def:bincsst}.
A generalization to arbitrary fields was proposed and studied in~\cite{alberto}.
In this paper, we propose a new definition of CSS-$T$ code over binary field extensions, which is different from the one proposed in~\cite{alberto} but which finds a precise physical foundation in the transversal $T$-gate. We first give the definition and postpone its physical interpretation to Remark~\ref{rem:phys}.

\begin{definition}[$q$-ary CSS-$T$ code]
Let $\text{CSS}(C_1,C_2)$ be a CSS code defined over $\mathbb{F}_q^n$. Then it is CSS-$T$ if and only if the $T^{(\lambda)}$-gate preserves the code space for all $\lambda\in\mathbb{F}_q$, i.e.
\[(T^{(\lambda)})^{\otimes n}\ket{x}\in \text{CSS}(C_1,C_2) \quad \text{ for all } \ket x \in \text{CSS}(C_1,C_2).\]
\end{definition}


As already mentioned in Section~\ref{sec:qc}, the field trace  plays an important role in our construction since it appears in the definition of $\{T^{(\lambda)}\}_{\lambda\in\mathbb{F}_q}$.
We will need the following concepts from classical coding theory.

\begin{definition}[Trace code and subfield subcode]
   Let $C\subseteq \mathbb{F}_q^n$ be an $\mathbb F_q$-linear code. The \textit{trace code}  of $C$ with respect to $\mathbb{F}_p$ is 
    \begin{equation*}
        \text{tr}(C)=\{(\text{tr}(c_1),\ldots,\text{tr}(c_n)) \mid (c_1,\ldots, c_n)\in C\},
    \end{equation*}
    where $\text{tr}:\mathbb{F}_q \to \mathbb{F}_p$ is the trace map.
    The \textit{subfield subcode} of $C$ with respect to $\mathbb{F}_p$ is 
    \begin{equation*}
        C|_{\mathbb{F}_p} = C\cap\mathbb{F}_p^n =  \{c \in C \mid c=(c_1,\ldots, c_n)\in C, \, c_i \in \mathbb{F}_p\}.
    \end{equation*}
\end{definition}

    By linearity of the trace map, trace codes are linear codes. Furthermore, we have $\text{tr}(A)\subseteq\text{tr}(B)$ whenever $A\subseteq B$. 
    Note that the converse is not true. Consider for instance the codes $A\subsetneq B\subseteq \mathbb{F}_4^2$ over $\mathbb{F}_4=\mathbb{F}_2[\alpha]$, where $\alpha^2+\alpha+1$, with generator matrices
    \begin{equation*}
        G_A=\begin{pmatrix}1 & \alpha\end{pmatrix},\quad G_B=\begin{pmatrix}1 & 0\\0 & \alpha\end{pmatrix}.
    \end{equation*}
    Then $\text{tr}(A)=\text{tr}(B) = \mathbb F_2^2$, but $B$ is not contained in $A$.

The trace and subfield subcodes of a linear code are closely related via the following theorem by Delsarte.

\begin{theorem}[Delsarte Theorem; see~\cite{delsarte}]
    Let $C\subseteq \mathbb{F}_q^n$ be a linear code.
    We have 
    \begin{equation*}
        (C|_{\mathbb{F}_p})^\perp = \textnormal{tr}(C^\perp).
    \end{equation*}
\end{theorem}

If a code meets certain conditions with respect to the star product, its trace code coincides with its subfield subcode.


\begin{lemma}[Theorem 9 in \cite{galoisinvariance}]\label{lemma:consequencegaloisinvariance}
    Let $C\subseteq\mathbb{F}_{2^s}^n$ be a linear code. The following are equivalent:
    \begin{enumerate}[]
        \item  $C=C^{\star 2}$ (i.e. $C$ is Galois invariant),
        \item $\text{tr}(C) = C|_{\mathbb{F}_2}$,
        \item $\text{dim}_{\mathbb{F}_2}(\text{tr}(C)) = \text{dim}_{\mathbb{F}_q}(C)$.
    \end{enumerate}
\end{lemma}

Property 1. in the previous lemma is often called \textit{Galois invariance}; see~\cite{galoisinvariance}.
Duality and trace are related as follows.

\begin{lemma}
    Let $C\subseteq\mathbb{F}_q^n$ be an $\mathbb F_q$-linear code. We have
    \begin{equation*}
        \text{tr}(C)^{\perp_{\mathbb{F}_p}}\subseteq \tr(C^{\perp_{\mathbb{F}_q}}).
    \end{equation*}
\end{lemma}
\begin{proof}
We omit the subscript $\F_p$
throughout this proof. First we prove that 
     \begin{align}\label{incl_perp}
         C^\perp|_{\mathbb{F}_p} \subseteq \left(C|_{\mathbb{F}_p}\right)^\perp.
     \end{align}
     Let $\smash{x\in C^\perp|_{\mathbb{F}_p}}$. Then $\langle x,y\rangle =0$ for all $y\in C$. In particular $\langle x, y\rangle =0 $ for all $\smash{y\in C|_{\mathbb{F}_p}}$, thus by definition $\smash{x\in\left(C|_{\mathbb{F}_p}\right)^\perp}$, proving inclusion \ref{incl_perp}.  
     Then, we can use Delsarte Theorem on $C$ and $D =C^\perp$, to find
     \begin{equation*}\left(\text{tr}(C^\perp)\right)^\perp = \left( \text{tr}(D)\right)^\perp = D^\perp|_{\mathbb{F}_p} \subseteq \left(D|_{\mathbb{F}_p}\right)^\perp=\text{tr}(D^\perp)=\text{tr}(C).
     \end{equation*}
     Taking the dual of this inclusion yields
     \begin{equation*}
         \text{tr}(C)^\perp \subseteq \text{tr}(C^\perp),
     \end{equation*}
     which concludes the proof.
\end{proof}

We are now ready to give a characterisation of $q$-ary CSS-$T$ pairs involving trace codes and star products.

\begin{theorem}[Characterisation of CSS-$T$ codes  over $\mathbb{F}_{2^s}$]\label{CSST_cond-q}
Let $(C_1, C_2)$ be a CSS pair with $C_2 \subseteq C_1\subseteq\mathbb{F}_{2^s}^n$. Then $(C_1,C_2)$ is CSS-$T$ if and only if 
\begin{align} \label{csst_char}
    \textnormal{tr}(C_1)\star \textnormal{tr}(C_1) \subseteq \textnormal{tr}(C_2)^\perp,
\end{align}
 where $\textnormal{tr}=\textnormal{tr}_{\mathbb{F}_{2^s}/\mathbb{F}_2}$ is the absolute trace map from $\mathbb{F}_{2^s}$ to the base field $\mathbb{F}_2$. 
\end{theorem}
\begin{proof} 
Let $(C_1,C_2)$ be a CSS code with associated stabiliser group $S$. Then a CSS-$T$ code preserves the code space and the stabilisers under the action of an $n$-qudit $T$-gate $T^{\otimes n}$. The conjugation rules of the $q$-ary Pauli group under action of the $T$-gate yield:
    \begin{equation*}
        TX^{(\mu)} T^\dagger = \begin{cases}
        X^{(\mu)} &\text{ if }\, \text{tr}(\mu) = 0,\\
        \frac{1}{\sqrt{2}}\left(X^{(\mu)}+iX^{(\mu)}Z^{(1)}\right)&\text{ if }\,\text{tr}(\mu)=1.
        \end{cases}
    \end{equation*}
    If we take an arbitrary stabiliser $E(a,b)\in S$ in the Weyl-Heisenberg representation, then its conjugation  reads
    \begin{equation*}
        T^{\otimes n}E(a,b)(T^\dagger)^{\otimes n} = \frac{1}{2^{w_H(\text{tr}(a))/2}}\sum_{\text{tr}(y)\preceq \text{tr}(a)} (-1)^{\text{tr}(b)\text{tr}(y)^\top} E(a,b+\text{tr}(y)).
    \end{equation*}
    Following an argument similar to  \cite{css-t}, we can deduce two properties from the transversality of the $T$-gate:
    \begin{itemize}
        \item The trace code $\text{tr}(C_2)\subseteq\mathbb{F}_2^n$ is even, i.e. all of its codewords have even weight.
        \item We define $\mathcal{Z}_j = \{\text{tr}(z)\preceq \text{tr}(a_j) \mid z\in C_1, E(0,z)\in S_z\}$. Then, for all $\text{tr}(y)\preceq\text{tr}(a)$, we require $E(a,b+\text{tr}(y))\in S$ to also be a stabiliser. Since stabilisers form an abelian group, we have
        \begin{equation*}
            0 = \langle b,b+\text{tr}(y)\rangle = \langle b,b \rangle + \langle b,\text{tr}(y)\rangle = \langle b,\text{tr}(y)\rangle,
        \end{equation*}
        so that $\text{tr}(y)\in\mathcal{Z}_j^\perp.$ Clearly, $\mathcal{Z}_j \subseteq \text{tr}(C_1)$, so that $\text{tr}(C_1)^\perp \subseteq \mathcal{Z}_j^\perp \subseteq \mathcal{Z}_j \subseteq \text{tr}(C_1)$. Since $\text{tr}(C_1)$ is a binary dual-containing code of even length, we know it contains a self-dual code by Lemma~\ref{dualcontaining}.
    \end{itemize}
    We can now follow similar steps as we did for the binary CSS-$T$ characterisation \ref{theorem:binarycsst}. Let $x\in\text{tr}(C_2)$, $a\in \text{tr}(C_1)$ and $z\in C_{\text{tr}(x)}\subseteq \text{tr}(C_1)^\perp$, where $C_{\text{tr}(x)}$ is a self-dual code contained in the support of $\text{tr}(x).$ Then, $\langle a\star x,z \rangle = \langle a,x\star z \rangle = \langle a,z \rangle = 0.$ Then,
    \begin{equation*}
        \text{tr}(C_1)\star\text{tr}(C_1)\subseteq\text{tr}(C_2)^\perp. \qedhere
    \end{equation*}
\end{proof}
Equivalently, one can state that a code pair $(C_1,C_2)$ is CSS-$T$ if and only if $(\text{tr}(C_1),\text{tr}(C_2))$ is a binary CSS-$T$ code pair.

\begin{remark} \label{rem:phys}
We can use the CSS-$T$ characterisation of Theorem~\ref{CSST_cond-q} to illustrate the compatibility of our definition \ref{csst_char} with the Eastin-Knill Theorem \ref{theorem:eastin-knill}. More precisely, if $(C_1, C_2)$ is a $q$-ary CSS pair, then it cannot execute a transversal Hadamard gate and transversal $T$-gate simultaneously, as long as the code distance satisfies $d>2$, since this would imply that the universal gate set $\{H, T, \text{CNOT}\}$ can be implemented fully transversally. To see this, we note that $C_1=C_2^\perp$ implies a transversal Hadamard gate since it places $X$-type and $Z$-type stabilisers on equal footing. According to the binary CSS-$T$ characterisation, this implies that
    \begin{equation*}
        C_1^\perp \subsetneq C_1\quad\text{and}\quad C_1=C_1\star C_1.
    \end{equation*}
    Note that the first inclusion is a proper inclusion. If it were an equality, then $\text{dim}(C_1)=\text{dim}(C_2^\perp)=\text{dim}(C_2)$, and $k^\text{CSS}=0$ would imply a trivial CSS code. From \cite{randriambololona}, we know that
    \begin{equation*}
        \text{dim}(C_1^{\star 2})\geq\text{min}(n,\text{dim}(C_1)+d_1^\perp-2).
    \end{equation*}
    Since $C_1=C_1^{\star 2}$, this reduces to $d_1^\perp \leq 2.$ By self-orthogonality of the perp-code $C_1^\perp$, we know that $d_1\leq d_1^\perp \leq 2,$ and thus $d^\text{CSS-$T$}\ngtr 2.$ For non-binary codes, we obtain similar results, namely that 
    \begin{equation*}
        C_1^\perp \subseteq C_1\quad\text{and}\quad \text{tr}(C_1)=\text{tr}(C_1)\star \text{tr}(C_1),
    \end{equation*}
    giving us a restriction on the distance of the trace code: $d(\text{tr}(C_1^\perp))\leq 2.$ By Delsarte Theorem applied to $C_1^\perp$, we have that 
    \begin{equation*}
        \left(C_1^\perp|_{\mathbb{F}_2}\right)^\perp = \text{tr}(C_1^{\perp\perp})=\text{tr}(C_1).
    \end{equation*}
    Since $\text{tr}(C_1)=\text{tr}(C_1)^{\star 2}$, its dual-distance must be $\leq 2$: $d(\left(C_1^\perp|_{\mathbb{F}_2}\right)^{\perp\perp})=d\left(C_1^\perp|_{\mathbb{F}_2}\right)\leq 2.$ 
    Either there exists a non-zero codeword $c \in C_1^\perp|_{\mathbb{F}_2}$ with Hamming weight $1\leq\omega^\text{H}(c)\leq 2$, so that $d^\text{CSS-$T$}=d(C_1^\perp)\leq 2$, or there exists no such codeword: $C_1^\perp|_{\mathbb{F}_2}=\{\textbf{0}\}$. However, the latter would imply that $\text{tr}(C_1)=\mathbb{F}_2^n$ by Delsarte Theorem, and by the CSS-$T$ characterisation we obtain $\text{tr}(C_2)=\{\textbf{0}\}$, so $d^\text{CSS-$T$}=1$. We conclude that our characterisation is consistent with Eastin-Knill Theorem.
    \end{remark}

It is known that if $(C_1, C_2)$ is a binary CSS-$T$ pair, then $C_2$ is self-orthogonal. An important consequence of Theorem \ref{CSST_cond-q} is that this property is paralleled in the $q$-ary case.

\begin{corollary}\label{lemma:traceselforth}
If $(C_1,C_2)$ is a CSS-$T$ pair over $\F_{2^s}$, then $\text{tr}(C_2)$ is a self-orthogonal binary code.
\end{corollary}
\begin{proof}
    Using the CSS-$T$ characterisation we find
    \begin{equation*}
        \text{tr}(C_2)\subseteq \text{tr}(C_1)\subseteq \text{tr}(C_1)\star\text{tr}(C_1) \subseteq \text{tr}(C_2)^\perp,
    \end{equation*}
    where the second inclusion holds because $\text{tr}(C_1)$ is a binary code. Thus $\text{tr}(C_2)$ is a self-orthogonal binary code.
\end{proof}

When we define quantum codes over non-binary fields, Pauli operators obtain a superscript $\lambda\in\mathbb{F}_q$, as we have seen in Eq. (\ref{eq:ztrace}) and (\ref{eq:ttrace}). By virtue of being CSS, \textit{all} Pauli operator will always be transversal regardless of their superscript. Next, we demonstrate that this property also follows for the $T^{(\lambda)}$-gates by linearity of the trace.

\begin{proposition}
        If a CSS code executes the $T$-gate transversally, it executes all $T^{(\lambda)}$-gates transversally for $\lambda\in\mathbb{F}_q.$
\end{proposition}
\begin{proof}
    In similar fashion to Theorem \ref{CSST_cond-q}, we find that
    \begin{equation*}
        T^{(\lambda)}X^{(\mu)} T^{(\lambda),\dagger} = \begin{cases}
        X^{(\mu)} &\text{ if }\, \text{tr}(\mu\lambda) = 0,\\
        \frac{1}{\sqrt{2}}\left(X^{(\mu)}+iX^{(\mu)}Z^{(\lambda)}\right)&\text{ if }\,\text{tr}(\mu\lambda)=1,
        \end{cases}
    \end{equation*}
    which gives the conjugation rule
    \begin{equation*}
         \left(T^{(\lambda)}\right)^{\otimes n}E(a,b)\left(T^{(\lambda),\dagger}\right)^{\otimes n} = \frac{1}{2^{w_H(\text{tr}(\lambda a))/2}}\sum_{\text{tr}(y)\preceq \text{tr}(\lambda a)} (-1)^{\text{tr}(b)\text{tr}(y)^\top} E(a,b+\text{tr}(y)).
    \end{equation*}
    By the linearity of the trace code, if $a\in C\subseteq \mathbb{F}_q^n$ then $\lambda a\in C$ for all $\lambda\in\mathbb{F}_q$. Therefore, we find the same conditions for transversal execution of the $T$-gate.
\end{proof}

The evenness of $\tr{(C_2)}$ is a necessary condition for a CSS-$T$ pair. However, evenness of $\tr{(C_1)}$ leads to some remarkable behaviour.

\begin{theorem}\label{theorem:order4}
    Let $(C_1,C_2)$ be a $q$-ary CSS-$T$ pair. If $\textnormal{tr}(C_1)$ is an even code, then the application of a transversal $T$-gate (i.e. $T^{\otimes n}$) implements a logical operator of multiplicative order 4.
\end{theorem}
\begin{proof}
    We use the same commutative diagram from Theorem 9 in \cite{goodcsst}:
\[
\begin{tikzcd}
{\ket{u}_L} \arrow[r, "\text{Enc}"] \arrow[d, swap, "\text{Op}_L"] 
  & \displaystyle\sum_{v\in C_2} {\ket{v + uH}} \arrow[d, "T^{\otimes n}"] \\
{\ket{u'}_L} \arrow[r, "\text{Enc}"] 
  & \displaystyle{\sum_{v\in C_2} \ket{v + u'H} = \sum_{v\in C_2} T^{\otimes n} \ket{v + uH}}
\end{tikzcd}
\]
where $\text{Op}_L$ is the corresponding logical operator, $\text{Enc}$ is the (quantum) encoding, and $H$ is the parity check matrix having a basis of $C_3$ in its rows, where $C_1=C_2\oplus C_3$. Applying the diagram four times gives the following: 
\begin{equation*}
    (T^4)^{\otimes n}\ket{v+uH}=e^{i\pi\omega^H(\tr{(v+uH)})}\ket{v+uH}=\ket{v+uH},
\end{equation*}
because $\tr{(C_1)}$ is even. This implies that the right hand side of the commutative diagram is equal, i.e. $\sum_{v\in C_2}(T^4)^{\otimes n}\ket{v+uH}=\sum_{v\in C_2}\ket{v+uH}$ which implies that $\text{Op}_L^4$ is the logical identity operator.
\end{proof}

As a consequence of Theorem~\ref{theorem:order4}, such codes only implement either a specific set of logical gates, among which the logical $S$-gate, a logical $Z$-gate or the logical identity. Since the $Z$-gate and identity are always transversal, the additional algebraic structure required to form a CSS-T pair does not yield further transversal advantages in these case. However, the controlled-$S$ gate is an element of $\mathcal{K}^{(3)}$ and may therefore be supplied to the Clifford group for universal quantum computation.




\section{Properties and applications of CSS-T codes}\label{sec:applications}

\subsection{Comparisons}

As described earlier, in this paper we deviate from the definition of $q$-ary CSS-$T$ codes  given in \cite{alberto}. A natural question is 
 if the two definitions are related. We  provide counterexamples showing that neither definition implies the other, but there also exist CSS-$T$ codes that satisfy both. In the sequel we will refer to the definition given in \cite{alberto} as ``BCR'' and to our definition as ``CPR'' for brevity.

\begin{example}[BCR $\not\Rightarrow$ CPR]
   Let $\mathbb F_8=\mathbb F_2[\alpha]$, where $\alpha^3+\alpha+1=0$ and let $C_2\subseteq C_1\subseteq \F_8^4$ be generated by 
   \begin{align*}
       G_2= \begin{pmatrix}1&\alpha& \alpha^2& 1+\alpha+\alpha^2\end{pmatrix},  \qquad  G_1 = \begin{pmatrix}1&\alpha& \alpha^2& 1+\alpha+\alpha^2\\1 & 1 & 1 & 1\end{pmatrix},
    \end{align*}
    respectively.
    Observe that $C_2$ is even, $C_1$ is self-orthogonal, and that  all non-zero $x \in C_2$ have full support. It follows that
    \[\pi_{\sigma(x)}(C_1)=C_1 \subseteq C_1^\perp=\pi_{\sigma(x)}(C_1)^{\perp_x},\]
    thus $(C_1, C_2)$ is a CSS-$T$ pair of type BCR. However, $\text{tr}(C_2)$ has generator matrix
    \[\begin{pmatrix}
        1 & 0 & 0 & 1\\
        0 & 0 & 1 & 1\\
        0 & 1 & 0 & 1
    \end{pmatrix} \in \mathbb F_2^{3\times 4}\]
    and is not self-orthogonal. By Lemma~\ref{lemma:traceselforth}, $(C_1, C_2)$ is not a CSS-$T$ pair of type CPR.
\end{example}

\begin{example}[CPR $\not\Rightarrow$ BCR]
     Let $\mathbb F_{16}=\mathbb F_2[\alpha]$, where $\alpha^4+\alpha+1=0$, and let $G$ be the tri-orthogonal matrix 
     \begin{equation*}G=\begin{pmatrix} G_\text{odd}\\\hline G_\text{even}\end{pmatrix}=
         \begin{pmatrix}
0 & 0 & 0 & 0 & 1 & 1 & 1 & 1 & 1 & 1 & 0 & 0 & 0 & 0 & 1\\
\hline
1 & 0 & 0 & 0 & 1 & 1 & 1 & 0 & 0 & 0 & 0 & 1 & 1 & 1 & 1\\
0 & 1 & 0 & 0 & 1 & 0 & 0 & 1 & 1 & 0 & 1 & 0 & 1 & 1 & 1\\
0 & 0 & 1 & 0 & 0 & 1 & 0 & 1 & 0 & 1 & 1 & 1 & 0 & 1 & 1\\
0 & 0 & 0 & 1 & 0 & 0 & 1 & 0 & 1 & 1 & 1 & 1 & 1 & 0 & 1
         \end{pmatrix}\in\mathbb{F}_2^{5\times 15}.
     \end{equation*}
Let $D$ be the linear span of the rows of $G$ of even weight and take $C_2=D\otimes_{\F_2}\F_{16}$. Let 
\begin{equation*}
   x= (\begin{smallmatrix}1, & 1, & 1, & 1, & 1+\alpha, & 1+\alpha^2, & 1+\alpha^3, & \alpha+\alpha^2, & \alpha+\alpha^3, & \alpha^2+\alpha^3, & \alpha+\alpha^2+\alpha^3, & 1+\alpha^2+\alpha^3, & 1+\alpha+\alpha^3, & 1+\alpha +
    \alpha^2, & 1+\alpha+\alpha^2+\alpha^3\end{smallmatrix})
\end{equation*}
and consider $C_1=\langle x \rangle \oplus C_2$, so that 
$\text{tr}(C_2)=\text{rowsp}(G_\text{even})$ and $\text{tr}(C_1)=\text{rowsp}(G)$. Then $(C_1,C_2)$ is a CSS-$T$ pair by our definition, but it's not a CSS-$T$ pair of type BCR as at least one of the codewords in $C_2$ is not even, e.g. $x$
has length and Hamming weight 15.
\end{example}

\begin{example}[Codes satisfying both BCR and CPR]
    Let $\mathbb F_{4}=\mathbb F_2[\alpha]$, where $\alpha^2+\alpha+1=0$. Consider $C_2\subseteq C_1 \subseteq \F_4^6$ with generator matrices 
    \begin{equation*}
        G_2=\begin{pmatrix}
            1 & 1 & \alpha+1 & \alpha+1 & \alpha & \alpha
        \end{pmatrix}, \qquad G_1 = \begin{pmatrix}
            1 & 1 & \alpha+1 & \alpha+1 & \alpha & \alpha\\ 1 & 1 & 1 & 1 & 1 & 1
        \end{pmatrix}.
    \end{equation*}
 Then   $(C_1, C_2)$ is a CSS-$T$ pair that is both of type BCR and CPR.
\end{example}

\subsection{CSS-\texorpdfstring{$T$}{T} codes from generalised Reed-Muller codes}

Now we illustrate a code construction of CSS-$T$ codes using a generalised Reed-Muller code.  While our approach draws inspiration from prior work such as \cite{felice}, it deviates in a fundamental aspect: the family of traces of a generalised Reed-Muller code is generally not closed under the star product $\star$. This key distinction influences both the structure and the applicability of the resulting codes over $\mathbb{F}_q$.

\begin{definition}[Generalised Reed-Muller codes]
    Let $\mathbb{F}_q[x_1,\ldots,x_m]_{\leq r}$ be the ring of polynomials over $\mathbb{F}_q$ in $m$ variables, of degree less than or equal to $r.$ For any polynomial $p\in\mathbb{F}_q[x_1,\ldots,x_m]_{\leq r}$, we denote by $\text{ev}_{\mathbb{F}_q^m}(p)$ the vector obtained by evaluation of $p$ at points of $\mathbb{F}_q^m$ in a fixed order. The $r^\text{th}$-order Reed-Muller code of length $q^m$ is then then given by
    \begin{equation*}
        \text{GRM}_q(r,m)=\{\text{ev}_{\mathbb{F}_q^m}(p)\mid p\in\mathbb{F}_q[x_1,\ldots,x_m]_{\leq r}\}.
    \end{equation*}
    \end{definition}
    
Like standard Reed-Muller codes, generalised Reed-Muller codes satisfy the nested property $\text{GRM}_q(r_1,m)\subseteq\text{GRM}_q(r_2,m)$ if and only if $r_1\leq r_2$. Similary, the family of GRM codes is closed under duality and  satisfies $\text{GRM}_q(r,m)^\perp=\text{GRM}_q(m(q-1)-r-1,m)$. Furthermore, similarly to Reed-Muller codes, they are closed under the star product: \[\textnormal{GRM}_q(r_1,m)\star\textnormal{GRM}_q(r_2,m)=\textnormal{GRM}_q(r_1+r_2,m) \quad  \text{for } m\geq 2.\]

\begin{theorem}\label{theorem:grm}
    Let $\textnormal{GRM}_q(1,m)$ be a generalised Reed-Muller code of first order, for $q=2^s$. Then
    \begin{equation*}
        \textnormal{tr}(\textnormal{GRM}_q(1,m))=\textnormal{RM}(1,ms).
    \end{equation*}
\end{theorem}
\begin{proof}
    Let $\mathcal{P}^m_q$ be the set of coordinates in $\mathbb{F}_q^m$, sorted in lexicographic order. Then $\text{tr}(\mathcal{P}_q^m)=\mathcal{P}_2^m$, since $\text{tr}(\mathbb{F}_q)=\mathbb{F}_2$, and the linearity of the trace preserves the lexicographic order. Let $p\in\mathbb{F}_q[x_1,\ldots,x_m]$ be given by $p=p_0+\sum_{j=1}^m p_j x_j$, where $\{p_0, \ldots, p_j\}\in\mathbb{F}_q$. Let $z=(z_1,\ldots,z_m)\in\mathcal{P}_q^m$ be a set of coordinates. Then, $\text{tr}(p(z))\in\mathcal{P}_2^n$. Thus, $\textnormal{tr}(\textnormal{GRM}_q(1,m))\subseteq\textnormal{RM}(1,ms).$ To prove equality, we show that the code dimensions of either side of the inclusion are the same. We know that $\text{dim}(\textnormal{GRM}_q(1,m))=m+1$, with a code length of $q^m$. Then, the Reed-Muller code of length $\smash{2^{m'}}$ must satisfy $\smash{2^{m'}=q^m=2^{ms}}$, so that $\text{dim}(\textnormal{RM}(1,ms))=ms+1$, which is precisely the code dimension of the generalised Reed-Muller code under the image of the trace map. This concludes the proof. 
\end{proof}
\begin{remark}
    Let $\textnormal{GRM}_q(0,m)$ be a generalised Reed-Muller code of zeroth order, i.e. $\textnormal{GRM}_q(0,m) = \mathcal{R}_q^{q^m}$. 
    Then by linearity, we see that
    \begin{equation*}
        \textnormal{tr}(\textnormal{GRM}_q(0,m)) = \mathcal{R}_2^{2^{ms}}.
    \end{equation*}
\end{remark}

Then, only one family of non-trivial CSS-$T$ codes is available from generalised Reed-Muller codes.

\begin{theorem}
    If $(C_1,C_2)$ are generalised Reed-Muller codes such that they form a $2^s$-ary CSS-$T$ pair for $s>1$, then they must satisfy
    \begin{equation*}
        C_2 = \textnormal{GRM}_q(0, m)\quad\text{and}\quad C_1 =\textnormal{GRM}_q(1,m),
    \end{equation*}
    under the condition $ms\geq 3$.
\end{theorem}
\begin{proof}
    From Theorem \ref{theorem:grm} and its corollary, we find that generalised Reed-Muller codes only produce a binary Reed-Muller code under the trace map if and only if $r\in\{0,1\}$. Since non-trivial CSS codes must satisfy $\text{dim}(C_1)>\text{dim}(C_2)$, we strictly find $C_2 = \textnormal{GRM}_q(0, m)$ and $ C_1 =\textnormal{GRM}_q(1,m)$. Under the CSS-$T$ characterization  \ref{csst_char}, we find that 
    \begin{equation*}
        \text{RM}(1,ms)\star \text{RM}(1,ms)
        \subseteq
        \text{RM}(0,ms)^\perp= \text{RM}(ms-1,ms),  
    \end{equation*}
    so that we find $1+1\leq ms-1$, or $ms\geq 3$.
\end{proof}

Notice that this is in contrast with the binary case, where many more Reed-Muller pairs can be used to build CSS-$T$ codes.

\subsection{CSS-\texorpdfstring{$T$}{T} codes from cyclic codes} We now turn to an example of $q$-ary CSS-$T$ codes using cyclic codes, and characterise when two cyclic codes together form a CSS-$T$ pair.

\begin{definition}[Cyclic codes]
    Let $C\subseteq\mathbb{F}_q^n$ be a length-$n$ code such that $\text{gcd}(q,n)=1$. Let $\mathbb{Z}_n$ be the ring of integers modulo $n$. Let $g(x)\in\mathbb{F}_q[x]$ be the \textit{generator} polynomial of the code, such that $g(x) \mid x^n-1$. Let $\beta$ be a primitive $n$-th root of unity in some extension field of $\mathbb{F}_q$. We define
    \begin{enumerate}
        \item The \textit{defining set} $J=\{j\in \mathbb{Z}_n \mid g(\beta^j) = 0\}$,
        \item The \textit{generating set} $I=\{i\in \mathbb{Z}_n \mid g(\beta^i) \neq 0\}$.
    \end{enumerate}
    Note that $I\cup J = \mathbb{Z}_n.$ We also define $-I=\{-i\in\mathbb{Z}_n \mid i\in I\}.$

\end{definition}

\begin{lemma}[Cyclic trace codes]\label{lemma:tracegenerator}
    We highlight two properties of trace codes of cyclic codes:
    \begin{enumerate}
        \item The trace code $\text{tr}(C)$ of a cyclic code $C$ is also cyclic.
        \item If $C\subseteq\mathbb{F}_q$ is a cyclic code with generator polynomial $g(x)\in\mathbb{F}_q[x]$, then $\text{tr}(C)\subseteq\mathbb{F}_p$ is a cyclic code with generator polynomial $\eta(g(x))\in\mathbb{F}_p[x]$, where $\eta(g(x))$ is the unique largest-degree divisor of $g(x)$ with all coefficients in $\mathbb{F}_p$.
    \end{enumerate} 
\end{lemma}
\begin{proof}
    We prove both statements separately:
    \begin{enumerate}
    \item Let $\mathcal{T}$ be the cyclic shift operator on codewords $c\in C$, such that
    \begin{equation*}
        \mathcal{T}(c_1,c_2,\ldots,c_n)=(c_2,\ldots,c_n,c_1).
    \end{equation*}
    Then $\mathcal{T}\circ \text{tr}=\text{tr}\circ \mathcal{T}$ is apparent from
    \begin{align*}
        \mathcal{T}(\text{tr}(c_1,c_2,\ldots,c_n))&=(\text{tr}(c_2),\ldots,\text{tr}(c_n),\text{tr}(c_1))\\&= \text{tr}(c_2,\ldots,c_n,c_1)\\&=\text{tr}(\mathcal{T}(c_1,\ldots,c_n)).
    \end{align*}
    Thus for all $c\in\text{tr}(C)$, we have $\mathcal{T}(c)\in\text{tr}(C)$, so that $\text{tr}(C)$ is a cyclic code.
    \item See Lemma 2.1 in \cite{weirdcyclictrpaper}. \qedhere
    \end{enumerate} 
\end{proof}

\begin{example}
    Let $\mathbb{F}_4=\mathbb{F}_2[\alpha]$ with $\alpha^2+\alpha+1=0.$ Let $C\subseteq\mathbb{F}_4^9$ be the $[9,4,4]$-code generated by $g(x)=(x^3+\alpha)(x+\alpha)(x+\alpha+1)$, with generator matrix
    \begin{equation*}
        G = \begin{pmatrix}\alpha & \alpha & \alpha &1 & 1 & 1 &  0 & 0 & 0\\ 
        0 & \alpha & \alpha & \alpha & 1 & 1 & 1 & 0 & 0\\
        0 & 0 & \alpha & \alpha & \alpha & 1 & 1 & 1 & 0\\
        0 & 0 & 0 & \alpha & \alpha & \alpha & 1 & 1 & 1 \\
        \end{pmatrix}
    \end{equation*}
    Then $\text{tr}(C)\subseteq\mathbb{F}_2^9$ is the $[9,7,2]$-code with generator polynomial $\eta(g(x))=x^2+x+1$.  Its generator matrix is
    \begin{equation*}
        G = \begin{pmatrix}
1 & 1 & 1 & 0 & 0 & 0 & 0 & 0 & 0\\
0 & 1 & 1 & 1 & 0 & 0 & 0 & 0 & 0\\
0 & 0 & 1 & 1 & 1 & 0 & 0 & 0 & 0\\
0 & 0 & 0 & 1 & 1 & 1 & 0 & 0 & 0\\
0 & 0 & 0 & 0 & 1 & 1 & 1 & 0 & 0\\
0 & 0 & 0 & 0 & 0 & 1 & 1 & 1 & 0\\
0 & 0 & 0 & 0 & 0 & 0 & 1 & 1 & 1\\
        \end{pmatrix}.
    \end{equation*}
\end{example}
    The following lemma is a variant of Lemma \ref{lemma:tracegenerator}.
\begin{lemma}
    Let $C(I^{(q)})\subseteq\mathbb{F}_q^n$ be a cyclic code, such that $\text{gcd}(q,n)=1$ and $I^{(q)}$ is its generating set. Then the binary cyclic code $\text{tr}(C)$ has the generating set of cyclotomic cosets $$I^{(2)}=\{C_s \mid C_s\subseteq I^{(q)}, C_{2s}=C_s\}.$$ 
\end{lemma}

This leads us to the necessary and sufficient conditions for two classical cyclic codes over $\mathbb{F}_q$ to form a CSS-$T$ pair, which relies on the \textit{Minkowski sum}.
\begin{definition}[Minkowski sum]
    Given two cyclotomic cosets $I_1,I_2$, their \textit{(Minkowski) sum} is given by
    \begin{equation*}
        I_1+I_2=\{i_1+i_2 \mid i_1\in I_1, i_2\in I_2\}\subseteq\mathbb{Z}_n.
    \end{equation*}
\end{definition}

The resulting theorem follows, which is a generalization to the $q$-ary case of Theorem 4.8 in \cite{eduardo}.

\begin{theorem}
    Let $(C_1,C_2)$ be a $q$-ary CSS-$T$ pair of cyclic codes $C_1\left(I_1^{(q)}\right)$ and $C_2\left(I_2^{(q)}\right)$ of length~$n$. Then the following hold:
    \begin{enumerate}
        \item $I^{(q)}_2 \subseteq I^{(q)}_1$,
        \item $n\notin \left(I_1^{(2)}+I_1^{(2)}+I_2^{(2)}\right)$.
    \end{enumerate}
\end{theorem}
\begin{proof}
    (1) is a direct consequence of the CSS condition $C_2\subseteq C_1$. For cyclic codes, this is equivalent to stating that $g_2(x) \mid g_1(x)$, which is true if and only if $\smash{I_2^{(q)}\subseteq I_1^{(q)}}$. For (2), we use the fact that $\text{tr}(C(I^{(q)}))=C(I^{(2)})\subseteq \mathbb{F}_2^n.$ The $q$-ary identity yields:
    \begin{equation*}
        C\left(I_1^{(2)}\right)\star C\left(I_1^{(2)}\right)\subseteq\left(C\left(I_2^{(2)}\right)\right)^\perp\quad\Leftrightarrow\quad \mathcal{R}_q^n \subseteq C\left(I_1^{(2)}+I_1^{(2)}+I_2^{(2)}\right)^\perp.
    \end{equation*}
    It follows that $n\notin\left(I_1^{(2)}+I_1^{(2)}+I_2^{(2)}\right)$.
\end{proof}

\section{Bounds and asymptotically good CSS-T codes over binary extension fields}\label{sec:bounds}

In this section, we look at the asymptotic behaviour of CSS-$T$ codes. More in detail, we show that asymptotically good sequences of CSS-$T$ codes exist, and that they can be derived from sequences of CSS codes. First, we study some fundamental properties of CSS-$T$ codes.

\begin{proposition} 
    Let $(C_1,C_2)$ form a $q$-ary CSS-$T$ pair with parameters $[[n,k,d]]$, and let $R=k/n$ and $\delta=d/n$ denote the channel capacity and relative minimum distance respectively. Then the following statements hold true:
    \begin{enumerate}
        \item If there exists a codeword $x\in C_2$ with Hamming weight $\omega^\text{H}(x)\geq n+1-k_2$, then \[R+\frac{\delta}{2}\leq \frac{1}{2}.\]
         \item If there exists a codeword $x\in C_2$ with Hamming weight $\omega^\text{H}(x)\geq n-d_2$, then \[R+{\delta}\leq \frac{1}{2}+\frac{1}{n}.\]
          \item If there exists a codeword $x\in C_2$ with Hamming weight $\omega^\text{H}(x)\geq n-d_1$, then \[R+\frac{3}{2}\delta \leq \frac{1}{2}+\frac{2}{n}.\]
    \end{enumerate}
\end{proposition}
\begin{proof}
    These statements follow easily from Theorem 3.9 in Ref.~\cite{alberto}, where the binary case is analyzed, using the fact that $\text{dim}(\text{tr}(C))\geq \text{dim}(C)$.
\end{proof}

To prove that an infinite number of asymptotically good sequences of CSS-$T$ codes exist, we first recall a fundamental result by Panteleev and Kalachev.

\begin{theorem}[See~\cite{proofasymp}]\label{theorem:AGqary}
    For every $R\in(0,1)$ and finite field $\mathbb{F}_q$ there exists an explicit family of quantum LDPC codes over $\mathbb{F}_q$ with parameters $[[n,k\geq Rn, d=\Theta(n)]]_q$ as $n\to\infty$.
\end{theorem}

In similar fashion as~\cite{goodcsst}, we can produce CSS-$T$ pairs over $\mathbb{F}_{2^s}^{2n}$ from any CSS pair over $\mathbb{F}_2^n$. Let $C\subseteq\mathbb{F}_q^n$ be a linear code and let $\phi:C\rightarrow{}\mathbb{F}_q^n$ be a linear map. We can extend the code as follows:
\begin{equation*}
    C^{\phi}=\{(x,\phi(x))\mid x\in C\}\subseteq\mathbb{F}_q^{2n}.
\end{equation*}
 In \cite{goodcsst} the authors give a necessary and sufficient condition on the map $\phi$ for $(C_1^\phi, C_2^\phi)$ to be a CSS-$T$ pair, provided that $(C_1, C_2)$ is CSS. We prove  that combining their result with the absolute trace gives a similar condition for CSS-$T$ pairs over a field of characteristic~$2$. 

\begin{theorem}
    Let $(C_1,C_2)$ be a \textnormal{CSS} code over $\mathbb{F}_{2^s}^n$ for some $s\in\mathbb{N}$. Then, the pair $(C_1^\phi, C_2^\phi)$ is a \textnormal{CSS-$T$} pair if and only if $\phi$ satisfies
    \begin{equation*}
        \omega^{\text{H}}(\textnormal{tr}(x)\star \textnormal{tr}(y)\star \textnormal{tr}(z))+\omega^\text{H}(\tr(\phi(x)) \star \tr(\phi(y))\star\tr(\phi(z))=0 \mod 2.
    \end{equation*}
    for all $x, y \in C_1$, $z\in C_2$.
\end{theorem}
\begin{proof}
   Suppose $(C_1^\phi,C_2^\phi)$ is a CSS-$T$ pair. By Theorem \ref{CSST_cond-q}, this is equivalent to 
    \[\text{tr}(C_1^\phi)\star \text{tr}(C_1^\phi)\subseteq\left(\text{tr}(C_2^\phi)\right)^\perp, \]
    i.e., for all $x,y \in C_1$, $z\in C_2$, 
    \begin{align*}
        0 & =\langle \tr(x) \star \tr(y), \tr(z) \rangle + \langle \tr(\phi(x))\star\tr(\phi(y)), \tr(\phi(z) \rangle   \\
        & = \sum_{i=1}^n{\tr(x_i)\tr(y_i)\tr(z_i)}+\sum_{i=1}^n{\tr(\phi(x)_i)\tr(\phi(y)_i)\tr(\phi(z)_i)}.
    \end{align*}   
    Since these are binary vectors, this is equivalent to 
    \begin{equation*}
        \omega^\text{H}(\textnormal{tr}(x)\star \textnormal{tr}(y)\star \textnormal{tr}(z))+\omega^\text{H}(\tr(\phi(x))\star \tr(\phi(y))\star \tr(\phi(z))=0. \qedhere
    \end{equation*}
\end{proof}
\begin{remark}
   The condition on $\phi$ is clearly satisfied when $\phi: C_1\hookrightarrow \mathbb F_q^n$ is the canonical embedding.
\end{remark}

It is not clear whether the identity map is the optimal map with respect to the parameters $[[n,k,d]]$. It remains an open question which map $\phi$ yields the best code parameters. This length-doubling procedure produces even $\text{tr}(C_1)$-codes, therefore implementing the logical identity operator. It remains to be seen whether we can create CSS-$T$ codes that map $T^{\otimes n}$ to logical $T$-gates via some map $\phi$.

Nevertheless, using the identity map allows us to produce an asymptotically good sequence of CSS-$T$ codes.

\begin{corollary}
    There exist asymptotically good sequences of CSS-$T$ codes over $\F_{2^s}$ for any $s\ge 1$.
\end{corollary}
\begin{proof}
    If $(C_1,C_2)$ is a CSS pair with parameters $[[n,k,d]]$ over $\F_{2^s}$, then using the map $\phi=\text{id}$ to extend it to a CSS-$T$ code provides a code with parameters $[[2n,k,\geq d]]$. A sequence of CSS pairs with rate and relative minimum distance
    \begin{equation*}
    \rho^\text{CSS}=\limsup_{n\to\infty} \frac{k}{n} >0,\quad\delta^\text{CSS}=\limsup_{n\to\infty} \frac{d}{n} >0
    \end{equation*}
    can therefore be transformed into a sequence of CSS-$T$ pairs with rate and relative minimum distance
    \begin{equation*}
    \rho^\text{CSS-$T$}=\limsup_{n\to\infty} \frac{k}{2n} = \frac{\rho}{2} >0,\quad\delta^\text{CSS-$T$} \geq \limsup_{n\to\infty} \frac{d}{2n} =\frac{\delta}{2} > 0.
    \end{equation*}
    Let $\{C_j^\text{CSS}\subseteq\mathbb{F}_{2^s}^{n_j}\}_{j\in\mathbb{N}}$ be a sequence of asymptotically good LDPC CSS codes with asymptotic rate $R\in (0,1)$, whose existence is guaranteed by Theorem \ref{theorem:AGqary}. Then using the lengthening procedure induced by $\phi$, we produce an asymptotically good sequence of LDPC CSS-$T$ codes $\{C_j^\text{CSS-$T$}\subseteq\mathbb{F}_{2^s}^{2n_j}\}_{j\in\mathbb{N}}$ of asymptotic rate $R'\in(0,\frac{1}{2})$.
\end{proof}

\bibliography{Bibliography}

\begin{thebibliography}{38}%
\makeatletter
\providecommand \@ifxundefined [1]{%
 \@ifx{#1\undefined}
}%
\providecommand \@ifnum [1]{%
 \ifnum #1\expandafter \@firstoftwo
 \else \expandafter \@secondoftwo
 \fi
}%
\providecommand \@ifx [1]{%
 \ifx #1\expandafter \@firstoftwo
 \else \expandafter \@secondoftwo
 \fi
}%
\providecommand \natexlab [1]{#1}%
\providecommand \enquote  [1]{``#1''}%
\providecommand \bibnamefont  [1]{#1}%
\providecommand \bibfnamefont [1]{#1}%
\providecommand \citenamefont [1]{#1}%
\providecommand \href@noop [0]{\@secondoftwo}%
\providecommand \href [0]{\begingroup \@sanitize@url \@href}%
\providecommand \@href[1]{\@@startlink{#1}\@@href}%
\providecommand \@@href[1]{\endgroup#1\@@endlink}%
\providecommand \@sanitize@url [0]{\catcode `\\12\catcode `\$12\catcode `\&12\catcode `\#12\catcode `\^12\catcode `\_12\catcode `\%12\relax}%
\providecommand \@@startlink[1]{}%
\providecommand \@@endlink[0]{}%
\providecommand \url  [0]{\begingroup\@sanitize@url \@url }%
\providecommand \@url [1]{\endgroup\@href {#1}{\urlprefix }}%
\providecommand \urlprefix  [0]{URL }%
\providecommand \Eprint [0]{\href }%
\providecommand \doibase [0]{https://doi.org/}%
\providecommand \selectlanguage [0]{\@gobble}%
\providecommand \bibinfo  [0]{\@secondoftwo}%
\providecommand \bibfield  [0]{\@secondoftwo}%
\providecommand \translation [1]{[#1]}%
\providecommand \BibitemOpen [0]{}%
\providecommand \bibitemStop [0]{}%
\providecommand \bibitemNoStop [0]{.\EOS\space}%
\providecommand \EOS [0]{\spacefactor3000\relax}%
\providecommand \BibitemShut  [1]{\csname bibitem#1\endcsname}%
\let\auto@bib@innerbib\@empty
\bibitem [{\citenamefont {Shor}(1994)}]{primefac1}%
  \BibitemOpen
  \bibfield  {author} {\bibinfo {author} {\bibfnamefont {P.~W.}\ \bibnamefont {Shor}},\ }\bibfield  {title} {\bibinfo {title} {Algorithms for quantum computation: discrete logarithms and factoring},\ }in\ \href {https://doi.org/10.1109/SFCS.1994.365700} {\emph {\bibinfo {booktitle} {Proceedings 35th Annual Symposium on Foundations of Computer Science}}}\ (\bibinfo {year} {1994})\ pp.\ \bibinfo {pages} {124--134}\BibitemShut {NoStop}%
\bibitem [{\citenamefont {Shor}(1997)}]{primefac2}%
  \BibitemOpen
  \bibfield  {author} {\bibinfo {author} {\bibfnamefont {P.~W.}\ \bibnamefont {Shor}},\ }\bibfield  {title} {\bibinfo {title} {Polynomial-time algorithms for prime factorization and discrete logarithms on a quantum computer},\ }\href {https://doi.org/10.1137/s0097539795293172} {\bibfield  {journal} {\bibinfo  {journal} {SIAM Journal on Computing}\ }\textbf {\bibinfo {volume} {26}},\ \bibinfo {pages} {1484–1509} (\bibinfo {year} {1997})}\BibitemShut {NoStop}%
\bibitem [{\citenamefont {Grover}(1996)}]{grover}%
  \BibitemOpen
  \bibfield  {author} {\bibinfo {author} {\bibfnamefont {L.~K.}\ \bibnamefont {Grover}},\ }\bibfield  {title} {\bibinfo {title} {A fast quantum mechanical algorithm for database search},\ }in\ \href {https://doi.org/10.1145/237814.237866} {\emph {\bibinfo {booktitle} {Proceedings of the Twenty-Eighth Annual ACM Symposium on Theory of Computing}}},\ \bibinfo {series and number} {STOC '96}\ (\bibinfo  {publisher} {Association for Computing Machinery},\ \bibinfo {address} {New York, NY, USA},\ \bibinfo {year} {1996})\ p.\ \bibinfo {pages} {212–219}\BibitemShut {NoStop}%
\bibitem [{\citenamefont {Temme}\ \emph {et~al.}(2017)\citenamefont {Temme}, \citenamefont {Bravyi},\ and\ \citenamefont {Gambetta}}]{zne}%
  \BibitemOpen
  \bibfield  {author} {\bibinfo {author} {\bibfnamefont {K.}~\bibnamefont {Temme}}, \bibinfo {author} {\bibfnamefont {S.}~\bibnamefont {Bravyi}},\ and\ \bibinfo {author} {\bibfnamefont {J.~M.}\ \bibnamefont {Gambetta}},\ }\bibfield  {title} {\bibinfo {title} {Error mitigation for short-depth quantum circuits},\ }\bibfield  {journal} {\bibinfo  {journal} {Physical Review Letters}\ }\textbf {\bibinfo {volume} {119}},\ \href {https://doi.org/10.1103/physrevlett.119.180509} {10.1103/physrevlett.119.180509} (\bibinfo {year} {2017})\BibitemShut {NoStop}%
\bibitem [{\citenamefont {de~Keijzer}\ \emph {et~al.}(2025)\citenamefont {de~Keijzer}, \citenamefont {Visser}, \citenamefont {Tse},\ and\ \citenamefont {Kokkelmans.}}]{vqocsse}%
  \BibitemOpen
  \bibfield  {author} {\bibinfo {author} {\bibfnamefont {R.}~\bibnamefont {de~Keijzer}}, \bibinfo {author} {\bibfnamefont {L.}~\bibnamefont {Visser}}, \bibinfo {author} {\bibfnamefont {O.}~\bibnamefont {Tse}},\ and\ \bibinfo {author} {\bibfnamefont {S.}~\bibnamefont {Kokkelmans.}},\ }\bibfield  {title} {\bibinfo {title} {Fidelity-enhanced variational quantum optimal control},\ }\href {https://doi.org/10.1103/PhysRevA.111.052625} {\bibfield  {journal} {\bibinfo  {journal} {Phys. Rev. A}\ }\textbf {\bibinfo {volume} {111}},\ \bibinfo {pages} {052625} (\bibinfo {year} {2025})}\BibitemShut {NoStop}%
\bibitem [{\citenamefont {Huggins}\ \emph {et~al.}(2021)\citenamefont {Huggins}, \citenamefont {McArdle}, \citenamefont {O'Brien}, \citenamefont {Lee}, \citenamefont {Rubin}, \citenamefont {Boixo}, \citenamefont {Whaley}, \citenamefont {Babbush},\ and\ \citenamefont {McClean}}]{virtualdistillation}%
  \BibitemOpen
  \bibfield  {author} {\bibinfo {author} {\bibfnamefont {W.~J.}\ \bibnamefont {Huggins}}, \bibinfo {author} {\bibfnamefont {S.}~\bibnamefont {McArdle}}, \bibinfo {author} {\bibfnamefont {T.~E.}\ \bibnamefont {O'Brien}}, \bibinfo {author} {\bibfnamefont {J.}~\bibnamefont {Lee}}, \bibinfo {author} {\bibfnamefont {N.~C.}\ \bibnamefont {Rubin}}, \bibinfo {author} {\bibfnamefont {S.}~\bibnamefont {Boixo}}, \bibinfo {author} {\bibfnamefont {K.~B.}\ \bibnamefont {Whaley}}, \bibinfo {author} {\bibfnamefont {R.}~\bibnamefont {Babbush}},\ and\ \bibinfo {author} {\bibfnamefont {J.~R.}\ \bibnamefont {McClean}},\ }\bibfield  {title} {\bibinfo {title} {Virtual distillation for quantum error mitigation},\ }\href {https://doi.org/10.1103/PhysRevX.11.041036} {\bibfield  {journal} {\bibinfo  {journal} {Phys. Rev. X}\ }\textbf {\bibinfo {volume} {11}},\ \bibinfo {pages} {041036} (\bibinfo {year} {2021})}\BibitemShut {NoStop}%
\bibitem [{\citenamefont {Pelofske}\ and\ \citenamefont {Russo}(2025)}]{digitalzne}%
  \BibitemOpen
  \bibfield  {author} {\bibinfo {author} {\bibfnamefont {E.}~\bibnamefont {Pelofske}}\ and\ \bibinfo {author} {\bibfnamefont {V.}~\bibnamefont {Russo}},\ }\href {https://arxiv.org/abs/2503.06341} {\bibinfo {title} {Digital zero-noise extrapolation with quantum circuit unoptimization}} (\bibinfo {year} {2025}),\ \Eprint {https://arxiv.org/abs/2503.06341} {arXiv:2503.06341 [quant-ph]} \BibitemShut {NoStop}%
\bibitem [{\citenamefont {Cai}\ \emph {et~al.}(2023)\citenamefont {Cai}, \citenamefont {Babbush}, \citenamefont {Benjamin}, \citenamefont {Endo}, \citenamefont {Huggins}, \citenamefont {Li}, \citenamefont {McClean},\ and\ \citenamefont {O’Brien}}]{review}%
  \BibitemOpen
  \bibfield  {author} {\bibinfo {author} {\bibfnamefont {Z.}~\bibnamefont {Cai}}, \bibinfo {author} {\bibfnamefont {R.}~\bibnamefont {Babbush}}, \bibinfo {author} {\bibfnamefont {S.~C.}\ \bibnamefont {Benjamin}}, \bibinfo {author} {\bibfnamefont {S.}~\bibnamefont {Endo}}, \bibinfo {author} {\bibfnamefont {W.~J.}\ \bibnamefont {Huggins}}, \bibinfo {author} {\bibfnamefont {Y.}~\bibnamefont {Li}}, \bibinfo {author} {\bibfnamefont {J.~R.}\ \bibnamefont {McClean}},\ and\ \bibinfo {author} {\bibfnamefont {T.~E.}\ \bibnamefont {O’Brien}},\ }\bibfield  {title} {\bibinfo {title} {Quantum error mitigation},\ }\bibfield  {journal} {\bibinfo  {journal} {Reviews of Modern Physics}\ }\textbf {\bibinfo {volume} {95}},\ \href {https://doi.org/10.1103/revmodphys.95.045005} {10.1103/revmodphys.95.045005} (\bibinfo {year} {2023})\BibitemShut {NoStop}%
\bibitem [{\citenamefont {Gottesman}(1997)}]{gottesmanphd}%
  \BibitemOpen
  \bibfield  {author} {\bibinfo {author} {\bibfnamefont {D.}~\bibnamefont {Gottesman}},\ }\href {https://arxiv.org/abs/quant-ph/9705052} {\bibinfo {title} {Stabilizer codes and quantum error correction}} (\bibinfo {year} {1997}),\ \Eprint {https://arxiv.org/abs/quant-ph/9705052} {arXiv:quant-ph/9705052 [quant-ph]} \BibitemShut {NoStop}%
\bibitem [{\citenamefont {Campbell}(2021)}]{errorfigure}%
  \BibitemOpen
  \bibfield  {author} {\bibinfo {author} {\bibfnamefont {E.~T.}\ \bibnamefont {Campbell}},\ }\bibfield  {title} {\bibinfo {title} {Early fault-tolerant simulations of the {H}ubbard model},\ }\href@noop {} {\bibfield  {journal} {\bibinfo  {journal} {Quantum Science and Technology}\ }\textbf {\bibinfo {volume} {7}},\ \bibinfo {pages} {015007} (\bibinfo {year} {2021})}\BibitemShut {NoStop}%
\bibitem [{\citenamefont {Kivlichan}\ \emph {et~al.}(2020)\citenamefont {Kivlichan}, \citenamefont {Gidney}, \citenamefont {Berry}, \citenamefont {Wiebe}, \citenamefont {McClean}, \citenamefont {Sun}, \citenamefont {Jiang}, \citenamefont {Rubin}, \citenamefont {Fowler}, \citenamefont {Aspuru-Guzik} \emph {et~al.}}]{errorfigure2}%
  \BibitemOpen
  \bibfield  {author} {\bibinfo {author} {\bibfnamefont {I.~D.}\ \bibnamefont {Kivlichan}}, \bibinfo {author} {\bibfnamefont {C.}~\bibnamefont {Gidney}}, \bibinfo {author} {\bibfnamefont {D.~W.}\ \bibnamefont {Berry}}, \bibinfo {author} {\bibfnamefont {N.}~\bibnamefont {Wiebe}}, \bibinfo {author} {\bibfnamefont {J.}~\bibnamefont {McClean}}, \bibinfo {author} {\bibfnamefont {W.}~\bibnamefont {Sun}}, \bibinfo {author} {\bibfnamefont {Z.}~\bibnamefont {Jiang}}, \bibinfo {author} {\bibfnamefont {N.}~\bibnamefont {Rubin}}, \bibinfo {author} {\bibfnamefont {A.}~\bibnamefont {Fowler}}, \bibinfo {author} {\bibfnamefont {A.}~\bibnamefont {Aspuru-Guzik}}, \emph {et~al.},\ }\bibfield  {title} {\bibinfo {title} {Improved fault-tolerant quantum simulation of condensed-phase correlated electrons via trotterization},\ }\href@noop {} {\bibfield  {journal} {\bibinfo  {journal} {Quantum}\ }\textbf {\bibinfo {volume} {4}},\ \bibinfo {pages} {296} (\bibinfo {year} {2020})}\BibitemShut {NoStop}%
\bibitem [{\citenamefont {Calderbank}\ and\ \citenamefont {Shor}(1996)}]{CS}%
  \BibitemOpen
  \bibfield  {author} {\bibinfo {author} {\bibfnamefont {A.~R.}\ \bibnamefont {Calderbank}}\ and\ \bibinfo {author} {\bibfnamefont {P.~W.}\ \bibnamefont {Shor}},\ }\bibfield  {title} {\bibinfo {title} {Good quantum error-correcting codes exist},\ }\href {https://doi.org/10.1103/physreva.54.1098} {\bibfield  {journal} {\bibinfo  {journal} {Physical Review A}\ }\textbf {\bibinfo {volume} {54}},\ \bibinfo {pages} {1098–1105} (\bibinfo {year} {1996})}\BibitemShut {NoStop}%
\bibitem [{\citenamefont {Steane}(1996)}]{S}%
  \BibitemOpen
  \bibfield  {author} {\bibinfo {author} {\bibfnamefont {A.}~\bibnamefont {Steane}},\ }\bibfield  {title} {\bibinfo {title} {Multiple-particle interference and quantum error correction},\ }\href {https://doi.org/10.1098/rspa.1996.0136} {\bibfield  {journal} {\bibinfo  {journal} {Proceedings of the Royal Society of London. Series A: Mathematical, Physical and Engineering Sciences}\ }\textbf {\bibinfo {volume} {452}},\ \bibinfo {pages} {2551–2577} (\bibinfo {year} {1996})}\BibitemShut {NoStop}%
\bibitem [{\citenamefont {MacKay}\ \emph {et~al.}(2004)\citenamefont {MacKay}, \citenamefont {Mitchison},\ and\ \citenamefont {McFadden}}]{bicycleansatz}%
  \BibitemOpen
  \bibfield  {author} {\bibinfo {author} {\bibfnamefont {D.~J.~C.}\ \bibnamefont {MacKay}}, \bibinfo {author} {\bibfnamefont {G.}~\bibnamefont {Mitchison}},\ and\ \bibinfo {author} {\bibfnamefont {P.~L.}\ \bibnamefont {McFadden}},\ }\bibfield  {title} {\bibinfo {title} {Sparse-graph codes for quantum error correction},\ }\href {https://doi.org/10.1109/TIT.2004.834737} {\bibfield  {journal} {\bibinfo  {journal} {IEEE Transactions on Information Theory}\ }\textbf {\bibinfo {volume} {50}},\ \bibinfo {pages} {2315} (\bibinfo {year} {2004})}\BibitemShut {NoStop}%
\bibitem [{\citenamefont {Kovalev}\ and\ \citenamefont {Pryadko}(2013)}]{generalisedbicycleansatz}%
  \BibitemOpen
  \bibfield  {author} {\bibinfo {author} {\bibfnamefont {A.~A.}\ \bibnamefont {Kovalev}}\ and\ \bibinfo {author} {\bibfnamefont {L.~P.}\ \bibnamefont {Pryadko}},\ }\bibfield  {title} {\bibinfo {title} {Quantum {K}ronecker sum-product low-density parity-check codes with finite rate},\ }\href {https://doi.org/10.1103/PhysRevA.88.012311} {\bibfield  {journal} {\bibinfo  {journal} {Phys. Rev. A}\ }\textbf {\bibinfo {volume} {88}},\ \bibinfo {pages} {012311} (\bibinfo {year} {2013})}\BibitemShut {NoStop}%
\bibitem [{\citenamefont {Bravyi}\ \emph {et~al.}(2024)\citenamefont {Bravyi}, \citenamefont {Cross}, \citenamefont {Gambetta}, \citenamefont {Maslov}, \citenamefont {Rall},\ and\ \citenamefont {Yoder}}]{bravyi}%
  \BibitemOpen
  \bibfield  {author} {\bibinfo {author} {\bibfnamefont {S.}~\bibnamefont {Bravyi}}, \bibinfo {author} {\bibfnamefont {A.~W.}\ \bibnamefont {Cross}}, \bibinfo {author} {\bibfnamefont {J.~M.}\ \bibnamefont {Gambetta}}, \bibinfo {author} {\bibfnamefont {D.}~\bibnamefont {Maslov}}, \bibinfo {author} {\bibfnamefont {P.}~\bibnamefont {Rall}},\ and\ \bibinfo {author} {\bibfnamefont {T.~J.}\ \bibnamefont {Yoder}},\ }\bibfield  {title} {\bibinfo {title} {High-threshold and low-overhead fault-tolerant quantum memory},\ }\href {https://doi.org/10.1038/s41586-024-07107-7} {\bibfield  {journal} {\bibinfo  {journal} {Nature}\ }\textbf {\bibinfo {volume} {627}},\ \bibinfo {pages} {778–782} (\bibinfo {year} {2024})}\BibitemShut {NoStop}%
\bibitem [{\citenamefont {Boykin}\ \emph {et~al.}(2000)\citenamefont {Boykin}, \citenamefont {Mor}, \citenamefont {Pulver}, \citenamefont {Roychowdhury},\ and\ \citenamefont {Vatan}}]{universalgateset}%
  \BibitemOpen
  \bibfield  {author} {\bibinfo {author} {\bibfnamefont {P.~O.}\ \bibnamefont {Boykin}}, \bibinfo {author} {\bibfnamefont {T.}~\bibnamefont {Mor}}, \bibinfo {author} {\bibfnamefont {M.}~\bibnamefont {Pulver}}, \bibinfo {author} {\bibfnamefont {V.}~\bibnamefont {Roychowdhury}},\ and\ \bibinfo {author} {\bibfnamefont {F.}~\bibnamefont {Vatan}},\ }\bibfield  {title} {\bibinfo {title} {A new universal and fault-tolerant quantum basis},\ }\href {https://doi.org/https://doi.org/10.1016/S0020-0190(00)00084-3} {\bibfield  {journal} {\bibinfo  {journal} {Information Processing Letters}\ }\textbf {\bibinfo {volume} {75}},\ \bibinfo {pages} {101} (\bibinfo {year} {2000})}\BibitemShut {NoStop}%
\bibitem [{\citenamefont {Eastin}\ and\ \citenamefont {Knill}(2009)}]{eastinknill}%
  \BibitemOpen
  \bibfield  {author} {\bibinfo {author} {\bibfnamefont {B.}~\bibnamefont {Eastin}}\ and\ \bibinfo {author} {\bibfnamefont {E.}~\bibnamefont {Knill}},\ }\bibfield  {title} {\bibinfo {title} {Restrictions on transversal encoded quantum gate sets},\ }\bibfield  {journal} {\bibinfo  {journal} {Physical Review Letters}\ }\textbf {\bibinfo {volume} {102}},\ \href {https://doi.org/10.1103/physrevlett.102.110502} {10.1103/physrevlett.102.110502} (\bibinfo {year} {2009})\BibitemShut {NoStop}%
\bibitem [{\citenamefont {Rengaswamy}\ \emph {et~al.}(2020{\natexlab{a}})\citenamefont {Rengaswamy}, \citenamefont {Calderbank}, \citenamefont {Newman},\ and\ \citenamefont {Pfister}}]{csstoriginal}%
  \BibitemOpen
  \bibfield  {author} {\bibinfo {author} {\bibfnamefont {N.}~\bibnamefont {Rengaswamy}}, \bibinfo {author} {\bibfnamefont {R.}~\bibnamefont {Calderbank}}, \bibinfo {author} {\bibfnamefont {M.}~\bibnamefont {Newman}},\ and\ \bibinfo {author} {\bibfnamefont {H.~D.}\ \bibnamefont {Pfister}},\ }\bibfield  {title} {\bibinfo {title} {Classical coding problem from transversal {$T$} gates},\ }in\ \href {https://doi.org/10.1109/isit44484.2020.9174408} {\emph {\bibinfo {booktitle} {2020 IEEE International Symposium on Information Theory (ISIT)}}}\ (\bibinfo  {publisher} {IEEE},\ \bibinfo {year} {2020})\ p.\ \bibinfo {pages} {1891–1896}\BibitemShut {NoStop}%
\bibitem [{\citenamefont {Rengaswamy}\ \emph {et~al.}(2020{\natexlab{b}})\citenamefont {Rengaswamy}, \citenamefont {Calderbank}, \citenamefont {Newman},\ and\ \citenamefont {Pfister}}]{css-t}%
  \BibitemOpen
  \bibfield  {author} {\bibinfo {author} {\bibfnamefont {N.}~\bibnamefont {Rengaswamy}}, \bibinfo {author} {\bibfnamefont {R.}~\bibnamefont {Calderbank}}, \bibinfo {author} {\bibfnamefont {M.}~\bibnamefont {Newman}},\ and\ \bibinfo {author} {\bibfnamefont {H.~D.}\ \bibnamefont {Pfister}},\ }\bibfield  {title} {\bibinfo {title} {On optimality of {CSS} codes for transversal {$T$}},\ }\href {https://doi.org/10.1109/jsait.2020.3012914} {\bibfield  {journal} {\bibinfo  {journal} {IEEE Journal on Selected Areas in Information Theory}\ }\textbf {\bibinfo {volume} {1}},\ \bibinfo {pages} {499–514} (\bibinfo {year} {2020}{\natexlab{b}})}\BibitemShut {NoStop}%
\bibitem [{\citenamefont {Panteleev}\ and\ \citenamefont {Kalachev}(2022)}]{proofasymp}%
  \BibitemOpen
  \bibfield  {author} {\bibinfo {author} {\bibfnamefont {P.}~\bibnamefont {Panteleev}}\ and\ \bibinfo {author} {\bibfnamefont {G.}~\bibnamefont {Kalachev}},\ }\href@noop {} {\bibinfo {title} {Asymptotically good quantum and locally testable classical {LDPC} codes}} (\bibinfo {year} {2022}),\ \Eprint {https://arxiv.org/abs/2111.03654} {arXiv:2111.03654 [cs.IT]} \BibitemShut {NoStop}%
\bibitem [{\citenamefont {Berardini}\ \emph {et~al.}(2024{\natexlab{a}})\citenamefont {Berardini}, \citenamefont {Dastbasteh}, \citenamefont {Martinez}, \citenamefont {Jain},\ and\ \citenamefont {Larrarte}}]{goodcsst}%
  \BibitemOpen
  \bibfield  {author} {\bibinfo {author} {\bibfnamefont {E.}~\bibnamefont {Berardini}}, \bibinfo {author} {\bibfnamefont {R.}~\bibnamefont {Dastbasteh}}, \bibinfo {author} {\bibfnamefont {J.~E.}\ \bibnamefont {Martinez}}, \bibinfo {author} {\bibfnamefont {S.}~\bibnamefont {Jain}},\ and\ \bibinfo {author} {\bibfnamefont {O.~S.}\ \bibnamefont {Larrarte}},\ }\bibfield  {title} {\bibinfo {title} {Asymptotically good {CSS-$T$} codes exist},\ }\href@noop {} {\bibfield  {journal} {\bibinfo  {journal} {arXiv preprint arXiv:2412.08586}\ } (\bibinfo {year} {2024}{\natexlab{a}})}\BibitemShut {NoStop}%
\bibitem [{\citenamefont {Berardini}\ \emph {et~al.}(2024{\natexlab{b}})\citenamefont {Berardini}, \citenamefont {Caminata},\ and\ \citenamefont {Ravagnani}}]{alberto}%
  \BibitemOpen
  \bibfield  {author} {\bibinfo {author} {\bibfnamefont {E.}~\bibnamefont {Berardini}}, \bibinfo {author} {\bibfnamefont {A.}~\bibnamefont {Caminata}},\ and\ \bibinfo {author} {\bibfnamefont {A.}~\bibnamefont {Ravagnani}},\ }\bibfield  {title} {\bibinfo {title} {Structure of {CSS} and {CSS-$T$} quantum codes},\ }\href {https://doi.org/10.1007/s10623-024-01415-9} {\bibfield  {journal} {\bibinfo  {journal} {Designs, Codes and Cryptography}\ }\textbf {\bibinfo {volume} {92}},\ \bibinfo {pages} {2801–2823} (\bibinfo {year} {2024}{\natexlab{b}})}\BibitemShut {NoStop}%
\bibitem [{\citenamefont {Gottesman}(1998)}]{gottesmanknill}%
  \BibitemOpen
  \bibfield  {author} {\bibinfo {author} {\bibfnamefont {D.}~\bibnamefont {Gottesman}},\ }\bibfield  {title} {\bibinfo {title} {The {H}eisenberg representation of quantum computers},\ }in\ \href {https://arxiv.org/abs/quant-ph/9807006} {\emph {\bibinfo {booktitle} {22nd International Colloquium on Group Theoretical Methods in Physics}}}\ (\bibinfo {year} {1998})\ \Eprint {https://arxiv.org/abs/quant-ph/9807006} {arXiv:quant-ph/9807006 [quant-ph]} \BibitemShut {NoStop}%
\bibitem [{\citenamefont {Kitaev}(1997)}]{solovay}%
  \BibitemOpen
  \bibfield  {author} {\bibinfo {author} {\bibfnamefont {A.~Y.}\ \bibnamefont {Kitaev}},\ }\bibfield  {title} {\bibinfo {title} {Quantum computations: algorithms and error correction},\ }\href {https://doi.org/10.1070/RM1997v052n06ABEH002155} {\bibfield  {journal} {\bibinfo  {journal} {Russian Mathematical Surveys}\ }\textbf {\bibinfo {volume} {52}},\ \bibinfo {pages} {1191} (\bibinfo {year} {1997})}\BibitemShut {NoStop}%
\bibitem [{\citenamefont {Aharonov}(2003)}]{toffolihadamard}%
  \BibitemOpen
  \bibfield  {author} {\bibinfo {author} {\bibfnamefont {D.}~\bibnamefont {Aharonov}},\ }\href {https://arxiv.org/abs/quant-ph/0301040} {\bibinfo {title} {A simple proof that {T}offoli and {H}adamard are quantum universal}} (\bibinfo {year} {2003}),\ \Eprint {https://arxiv.org/abs/quant-ph/0301040} {arXiv:quant-ph/0301040 [quant-ph]} \BibitemShut {NoStop}%
\bibitem [{\citenamefont {Nielsen}\ and\ \citenamefont {Chuang}(2010)}]{nielsenchuang}%
  \BibitemOpen
  \bibfield  {author} {\bibinfo {author} {\bibfnamefont {M.~A.}\ \bibnamefont {Nielsen}}\ and\ \bibinfo {author} {\bibfnamefont {I.~L.}\ \bibnamefont {Chuang}},\ }\href@noop {} {\emph {\bibinfo {title} {Quantum Computation and Quantum Information: 10th Anniversary Edition}}}\ (\bibinfo  {publisher} {Cambridge University Press},\ \bibinfo {year} {2010})\BibitemShut {NoStop}%
\bibitem [{\citenamefont {Litinski}(2019)}]{gameofsurfacecodes}%
  \BibitemOpen
  \bibfield  {author} {\bibinfo {author} {\bibfnamefont {D.}~\bibnamefont {Litinski}},\ }\bibfield  {title} {\bibinfo {title} {A game of surface codes: Large-scale quantum computing with lattice surgery},\ }\href {https://doi.org/10.22331/q-2019-03-05-128} {\bibfield  {journal} {\bibinfo  {journal} {Quantum}\ }\textbf {\bibinfo {volume} {3}},\ \bibinfo {pages} {128} (\bibinfo {year} {2019})}\BibitemShut {NoStop}%
\bibitem [{\citenamefont {Kitaev}(2003)}]{toric}%
  \BibitemOpen
  \bibfield  {author} {\bibinfo {author} {\bibfnamefont {A.}~\bibnamefont {Kitaev}},\ }\bibfield  {title} {\bibinfo {title} {Fault-tolerant quantum computation by anyons},\ }\href {https://doi.org/10.1016/s0003-4916(02)00018-0} {\bibfield  {journal} {\bibinfo  {journal} {Annals of Physics}\ }\textbf {\bibinfo {volume} {303}},\ \bibinfo {pages} {2–30} (\bibinfo {year} {2003})}\BibitemShut {NoStop}%
\bibitem [{\citenamefont {Landahl}\ \emph {et~al.}(2011)\citenamefont {Landahl}, \citenamefont {Anderson},\ and\ \citenamefont {Rice}}]{colorcodes}%
  \BibitemOpen
  \bibfield  {author} {\bibinfo {author} {\bibfnamefont {A.~J.}\ \bibnamefont {Landahl}}, \bibinfo {author} {\bibfnamefont {J.~T.}\ \bibnamefont {Anderson}},\ and\ \bibinfo {author} {\bibfnamefont {P.~R.}\ \bibnamefont {Rice}},\ }\href {https://arxiv.org/abs/1108.5738} {\bibinfo {title} {Fault-tolerant quantum computing with color codes}} (\bibinfo {year} {2011}),\ \Eprint {https://arxiv.org/abs/1108.5738} {arXiv:1108.5738 [quant-ph]} \BibitemShut {NoStop}%
\bibitem [{\citenamefont {Bravyi}\ and\ \citenamefont {Kitaev}(2005)}]{msd}%
  \BibitemOpen
  \bibfield  {author} {\bibinfo {author} {\bibfnamefont {S.}~\bibnamefont {Bravyi}}\ and\ \bibinfo {author} {\bibfnamefont {A.}~\bibnamefont {Kitaev}},\ }\bibfield  {title} {\bibinfo {title} {Universal quantum computation with ideal {C}lifford gates and noisy ancillas},\ }\bibfield  {journal} {\bibinfo  {journal} {Physical Review A}\ }\textbf {\bibinfo {volume} {71}},\ \href {https://doi.org/10.1103/physreva.71.022316} {10.1103/physreva.71.022316} (\bibinfo {year} {2005})\BibitemShut {NoStop}%
\bibitem [{\citenamefont {Bravyi}\ and\ \citenamefont {Haah}(2012)}]{bravyihaahoverhead}%
  \BibitemOpen
  \bibfield  {author} {\bibinfo {author} {\bibfnamefont {S.}~\bibnamefont {Bravyi}}\ and\ \bibinfo {author} {\bibfnamefont {J.}~\bibnamefont {Haah}},\ }\bibfield  {title} {\bibinfo {title} {Magic-state distillation with low overhead},\ }\bibfield  {journal} {\bibinfo  {journal} {Physical Review A}\ }\textbf {\bibinfo {volume} {86}},\ \href {https://doi.org/10.1103/physreva.86.052329} {10.1103/physreva.86.052329} (\bibinfo {year} {2012})\BibitemShut {NoStop}%
\bibitem [{\citenamefont {Andrade}\ \emph {et~al.}(2025)\citenamefont {Andrade}, \citenamefont {Bolkema}, \citenamefont {Dexter}, \citenamefont {Eggers}, \citenamefont {Luongo}, \citenamefont {Manganiello},\ and\ \citenamefont {Szramowski}}]{felice}%
  \BibitemOpen
  \bibfield  {author} {\bibinfo {author} {\bibfnamefont {E.}~\bibnamefont {Andrade}}, \bibinfo {author} {\bibfnamefont {J.}~\bibnamefont {Bolkema}}, \bibinfo {author} {\bibfnamefont {T.}~\bibnamefont {Dexter}}, \bibinfo {author} {\bibfnamefont {H.}~\bibnamefont {Eggers}}, \bibinfo {author} {\bibfnamefont {V.}~\bibnamefont {Luongo}}, \bibinfo {author} {\bibfnamefont {F.}~\bibnamefont {Manganiello}},\ and\ \bibinfo {author} {\bibfnamefont {L.}~\bibnamefont {Szramowski}},\ }\href {https://arxiv.org/abs/2305.06423} {\bibinfo {title} {{CSS-$T$} codes from {R}eed {M}uller codes}} (\bibinfo {year} {2025}),\ \Eprint {https://arxiv.org/abs/2305.06423} {arXiv:2305.06423 [cs.IT]} \BibitemShut {NoStop}%
\bibitem [{\citenamefont {Camps-Moreno}\ \emph {et~al.}(2024)\citenamefont {Camps-Moreno}, \citenamefont {López}, \citenamefont {Matthews}, \citenamefont {Ruano}, \citenamefont {San-José},\ and\ \citenamefont {Soprunov}}]{eduardo}%
  \BibitemOpen
  \bibfield  {author} {\bibinfo {author} {\bibfnamefont {E.}~\bibnamefont {Camps-Moreno}}, \bibinfo {author} {\bibfnamefont {H.~H.}\ \bibnamefont {López}}, \bibinfo {author} {\bibfnamefont {G.~L.}\ \bibnamefont {Matthews}}, \bibinfo {author} {\bibfnamefont {D.}~\bibnamefont {Ruano}}, \bibinfo {author} {\bibfnamefont {R.}~\bibnamefont {San-José}},\ and\ \bibinfo {author} {\bibfnamefont {I.}~\bibnamefont {Soprunov}},\ }\bibfield  {title} {\bibinfo {title} {An algebraic characterization of binary {CSS-$T$} codes and cyclic {CSS-$T$} codes for quantum fault tolerance},\ }\bibfield  {journal} {\bibinfo  {journal} {Quantum Information Processing}\ }\textbf {\bibinfo {volume} {23}},\ \href {https://doi.org/10.1007/s11128-024-04427-5} {10.1007/s11128-024-04427-5} (\bibinfo {year} {2024})\BibitemShut {NoStop}%
\bibitem [{\citenamefont {Delsarte}(2003)}]{delsarte}%
  \BibitemOpen
  \bibfield  {author} {\bibinfo {author} {\bibfnamefont {P.}~\bibnamefont {Delsarte}},\ }\bibfield  {title} {\bibinfo {title} {On subfield subcodes of modified {R}eed-{S}olomon codes},\ }\href@noop {} {\bibfield  {journal} {\bibinfo  {journal} {IEEE Transactions on Information Theory}\ }\textbf {\bibinfo {volume} {21}},\ \bibinfo {pages} {575} (\bibinfo {year} {2003})}\BibitemShut {NoStop}%
\bibitem [{\citenamefont {Giorgetti}\ and\ \citenamefont {Previtali}(2010)}]{galoisinvariance}%
  \BibitemOpen
  \bibfield  {author} {\bibinfo {author} {\bibfnamefont {M.}~\bibnamefont {Giorgetti}}\ and\ \bibinfo {author} {\bibfnamefont {A.}~\bibnamefont {Previtali}},\ }\bibfield  {title} {\bibinfo {title} {Galois invariance, trace codes and subfield subcodes},\ }\href {https://doi.org/https://doi.org/10.1016/j.ffa.2010.01.002} {\bibfield  {journal} {\bibinfo  {journal} {Finite Fields and Their Applications}\ }\textbf {\bibinfo {volume} {16}},\ \bibinfo {pages} {96} (\bibinfo {year} {2010})}\BibitemShut {NoStop}%
\bibitem [{\citenamefont {Randriambololona}(2013)}]{randriambololona}%
  \BibitemOpen
  \bibfield  {author} {\bibinfo {author} {\bibfnamefont {H.}~\bibnamefont {Randriambololona}},\ }\bibfield  {title} {\bibinfo {title} {An upper bound of {S}ingleton type for componentwise products of linear codes},\ }\href {https://doi.org/10.1109/tit.2013.2281145} {\bibfield  {journal} {\bibinfo  {journal} {IEEE Transactions on Information Theory}\ }\textbf {\bibinfo {volume} {59}},\ \bibinfo {pages} {7936–7939} (\bibinfo {year} {2013})}\BibitemShut {NoStop}%
\bibitem [{\citenamefont {Gao}\ and\ \citenamefont {Fu}(2011)}]{weirdcyclictrpaper}%
  \BibitemOpen
  \bibfield  {author} {\bibinfo {author} {\bibfnamefont {Z.-H.}\ \bibnamefont {Gao}}\ and\ \bibinfo {author} {\bibfnamefont {F.-W.}\ \bibnamefont {Fu}},\ }\bibfield  {title} {\bibinfo {title} {Linear recurring sequences and subfield subcodes},\ }in\ \href {https://doi.org/10.1109/IWSDA.2011.6159409} {\emph {\bibinfo {booktitle} {Proceedings of the Fifth International Workshop on Signal Design and Its Applications in Communications}}}\ (\bibinfo {year} {2011})\ pp.\ \bibinfo {pages} {142--145}\BibitemShut {NoStop}%
\end{thebibliography}%

\end{document}